\documentclass[a4paper,11pt]{article}

\usepackage{fullpage,latexsym,amsthm,amsmath,color,amssymb,url,hyperref}
\usepackage{tikz}
\usetikzlibrary{arrows,decorations.pathreplacing,shapes}

\usepackage{color}

\input{colordvi}

\newtheorem{lemma}{Lemma}
\newtheorem{claim}{Claim}
\newtheorem{proposition}{Proposition}

\newtheorem{definition}{Definition}
\newtheorem{theorem}{Theorem}
\newtheorem{conjecture}{Conjecture}

\newcommand{\cost}{\mathsf{cost}}
\newcommand{\ecost}{\mathsf{ecost}}

\newcommand{\cvs}{\mathsf{ctvs}}
\newcommand{\ecvs}{\mathsf{ectvs}}
\newcommand{\ctw}{\mathsf{ctw}}
\newcommand{\tw}{\mathsf{tw}}
\newcommand{\pw}{\mathsf{pw}}
\newcommand{\cpw}{\mathsf{cpw}}
\newcommand{\R}{\mathsf{R}}
\newcommand{\SPtree}{\textsf{SP}-tree}
\newcommand{\G}{\mathbf{G}}
\newcommand{\DP}{\mathsf{DP}}

\newenvironment{proofclaim}{
\noindent \emph{Proof of claim.}{}{}} {\hfill$\Diamond$\vspace{1em}}

\title{\textbf{A polynomial time algorithm to compute the connected treewidth of a series-parallel graph}\footnote{Research supported by ANR projects DEMOGRAPH (ANR-16-CE40-0028) and ESIGMA (ANR-17-CE23-0010) and the French-German Collaboration ANR/DFG Project UTMA ANR-20-CE92-0027.}}

\author{Guillaume Mescoff\thanks{\'Ecole Normale Sup\'erieure de Rennes (ENS Rennes).}
\and Christophe Paul\thanks{LIRMM, Univ Montpellier, CNRS, Montpellier, France.}  \and Dimitrios M. Thilikos$^{\dagger}$}

\date{\today}

\definecolor{Red}{rgb}{1, 0 ,0}
\definecolor{Blue}{rgb}{0, 0 ,1}

\usepackage{color}

\newcommand{\xchange}[1]{#1}

\hypersetup{
colorlinks = true,
linkcolor = black!30!blue,
citecolor = black!30!green
}

\date{}

\begin{document}

\maketitle

\begin{abstract}
\noindent It is well known that the \emph{treewidth} of a graph $G$ corresponds to the \emph{node search number} where a team of searchers is pursuing a fugitive that is \emph{lazy} and  \emph{invisible} (or alternatively is  \emph{agile} and \emph{visible}) and has the ability to move with  infinite speed via unguarded paths. 
Recently, \emph{monotone and connected node search strategies} have been considered. A search strategy is \emph{monotone} if it prevents the fugitive from pervading again areas from where he had been expelled  and is 
\emph{connected} if, at each step, the set of vertices that is or {\sl  has been} occupied by the searchers induces a connected subgraph of $G$. It has been shown that the corresponding connected and monotone search number of a graph $G$ can be expressed as the \emph{connected treewidth}, denoted $\mathbf{ctw}(G),$  that is defined as the minimum width of a rooted tree-decomposition $({{\cal X},T,r})$,  where the union of the bags corresponding to the nodes of a path of $T$ containing the root $r$ is connected in $G$. 
%\sed{Removed sentnece «Clearly we have that $\mathbf{tw}(G)\leqslant \mathbf{ctw}(G)$»}
%Clearly we have that $\mathbf{tw}(G)\leqslant \mathbf{ctw}(G)$. 
In this paper, we initiate the algorithmic study of connected treewidth. We design a $O(n^2\cdot\log n)$-time dynamic programming algorithm to compute the connected treewidth of biconnected series-parallel graphs. At the price of an extra $n$ factor in the running time, our algorithm generalizes to graphs of treewidth at most two.
\end{abstract}
%
%\tableofcontents
%\newpage

%---------------------------------------------------------------------------------------------------------------
%---------------------------------------------------------------------------------------------------------------
\medskip


\section{Introduction}

Since its introduction~\cite{BB72,Hal76,RS84}, the concept of \emph{treewidth} of a graph has led to major advancements in graph theory. 
The treewidth parameter, denoted by $\tw,$  is central to the design of efficient graph algorithms (see~\cite{Arn85} for a survey on early works in this direction and \cite{BK08} for a recent survey). According to the theorem  of Courcelle~\cite{Cou92}, properties expressible in $\mathsf{MSO}_2$ logic can be tested in parameterized linear time on graph of bounded treewidth. This result established  treewidth as an important structural parameter in the context of parameterized complexity~\cite{DF99}.  Treewidth can be defined in several equivalent ways, while the standard definition is by means of a tree-decomposition (see Section~\ref{sec:prelim}). 
During the last two decades, a number of width-parameters have been introduced as combinatorial 
variants or alternatives to treewidth (see~\cite{HOS08} for a survey on width parameters). This paper deals with the \emph{connected treewidth} parameter, a new variant of treewidth that is motivated by the study of connected search games in graphs.

A \emph{node search} game is a game opposing a group of searchers and a fugitive,  occupying the vertices of a graph. In a search strategy, one searcher may either be placed to or removed from a vertex and a winning search strategy is a sequence of moves of the searchers that eventually leads to the capture of the fugitive. 
The capture of the fugitive  happens when a searcher lands on the vertex occupied by the fugitive  while the fugitive  cannot escape along a searcher-free path. The cost of a search strategy is the maximum number of searchers simultaneously occupying vertices of the graph. 
 The \emph{clean territory} during some step of the game is the set of  
vertices from which the fugitive  has been, so far, expelled by the searchers’ strategy.
A strategy is \emph{monotone} if it guarantees that the fugitive  will not be able to visit  an already 
cleaned territory. Also, a strategy is {\em connected} if it guarantees,  that at any step, the clean territory 
induces a connected graph.
The fugitive  can be \emph{lazy} or \emph{agile}. Being lazy means that the fugitive  is staying at his position, as long as a searcher is not moving on that position. An agile fugitive  may move at any time regardless of the move of the searchers. Also a fugitive  can be   \emph{visible} or \emph{invisible} in the sense that the searcher’s strategy {\sl may} 
or {\sl may not} take into account the current position of the fugitive.
The node search number of  a graph (for some of the above variants) is the minimum cost of a winning search strategy.

It is well-known that the {\sl search number against an invisible and lazy fugitive }  is equal to treewidth plus one, while the same quantity 
is also equal to the {\sl search number against a visible and agile fugitive }~\cite{DKT97,ST93}.
On the other hand, if the fugitive  is  {\sl invisible and agile}, then the corresponding search number is 
equal to the parameter of pathwidth plus one~\cite{KP86,Moh90,KP85,Kin92}.
Moreover  all the aforementioned versions of the game are {\em monotone} in the sense that, for every search strategy, there is a monotone one with the same cost~\cite{Bie91,ST93}. In this paper, we deal only with monotone version node search.

Interestingly, the motivating story of one of the seminal papers~\cite{Par78} on graph searching was inspired by an earlier article of the Breisch in Southwestern Cavers Journal~\cite{Bre67}. This paper  concerned speleotopological explorations, which are, by essence, connected explorations as the searchers cannot ``teleport’’ to a vertex that is away from the current clean territory. 
The connectivity constraint was considered for the first time in~\cite{BFS03,BFFF12}, where 
the price of connectivity, defined as the ratio between monotone connected node search number (for some of the considered strategy variants) and the node search number, was first considered. In~\cite{FN06b,FN08}, it was proven that the price of connectivity for monotone node search against a visible and agile fugitive  (or, equivalently, an invisible and lazy one) is $\Theta(\log n)$. 
 Dereniowski~\cite{Der12} introduced the notion of \emph{connected pathwidth} of a graph, denoted $\cpw(G)$ and showed its equivalence with the monotone node search number against an invisible and agile fugitive. He proved that, for every graph $G,$  $\cpw(G)\leqslant 2\cdot \pw(G)+O(1),$  henceforth resolving the question of the price of connectivity  for a monotone node search against an invisible and agile fugitive. Extending the work of~\cite{Der12}, Adler et al.~\cite{APT19} recently introduced a definition of \emph{connected treewidth}, denoted hereafter $\ctw(G)$. They proved that, as for treewidth, the connected treewidth parameter is equivalent to the connected monotone node search number against an invisible and  lazy fugitive  and, as in the non-connected setting, the same holds for the visible and  agile case. Also, they proved that  connected treewidth 
 is equivalent to a connected variant of the tree vertex separation number (see Section~\ref{sec:prelim}). 
 
In this paper, we are interested in the problem of computing the connected treewidth of a graph. So far, very little is known on the algorithmic aspects of the connected treewidth and connected pathwidth parameters. 
First of all, checking, given a graph $G$ and an integer $k,$  whether 
the connected treewidth (or the connected pathwidth) is at most $k$ is an {\sf NP}-complete problem (see~Theorem~\ref{th_eqdcoml}).
This means that, in general, we may not expect any polynomial algorithm for computing connected treewidth.
Very recently, the problem of  deciding {\sf cpw}$(G)\leq k$ was shown to be fixed parameterized tractable~\cite{Kante2020linear}, improving an $n^{O(k^2)}$-time algorithm~\cite{DOR19}. An explanation for this lack of algorithmic results partially relies on the fact that, unlike pathwidth and treewidth, connected treewidth and connected pathwidth parameters do not enjoy nice combinatorial properties such as closeness under taking of minors.
% But this is not correct for connected pathwidth/treewidth.
%\sed{removed sentence «But this is not correct for connected pathwidth/treewidth.»}
Interestingly, both parameters are closed under contractions \cite{APT19}. However,  for graph contractions, there is no analogue of the algorithmic machinery developed in the context of graph minors  (see Section~\ref{klopath} for a more developed discussion).

We stress that connected treewidth may be arbitrarily bigger than treewidth. For instance, consider a 
complete binary tree of height $k$, introduce a new vertex and make it adjacent to all the leaves of the tree. As observed in \cite{APT19}, this graph has connected treewidth $k$ while it is a series-parallel graph that has treewidth 2.

The above motivated us to initiate the study of computing the connected treewidth on the class of series-parallel graphs. First introduced by Macmahon in 1982~\cite{Mac1892} (see also~\cite{RS42}),
 series-parallel graphs are essentially graphs of treewidth at most two~\cite{Duf65}. More precisely, a graph has treewidth at most two if and only if each of its biconnected components induces a series-parallel graph. The recursive construction, by means of series and parallel composition (see Section~\ref{sec:prelim}), of series-parallel graphs allows to solve a large number of \textsf{NP}-hard problems in polynomial (or even linear) time, see for example~\cite{BLW85,BLW87,Hasin1986Effic}. It follows that the class of series-parallel graphs, among others, forms a natural test bed for the existence of efficient graph algorithms~\cite{BPT92}. 

In this paper, we design a $O(n^2\cdot\log n)$-time algorithm to compute the connected treewidth of a biconnected series-parallel graph. The algorithm is extended to a $O(n^3\cdot\log n)$-algorithm for graphs of treewidth at most $2$. This result constitutes a first step toward the computation of connected treewidth parameterized by treewidth (see Section~\ref{klopath} for a discussion).

%---------------------------------------------------------------------------------------------------------------
%---------------------------------------------------------------------------------------------------------------
\section{Preliminaries}
\label{sec:prelim}

We use standard notations for graphs, as for example in~\cite{Die05}. We consider undirected and simple graphs. We let $G=(V,E)$ denote a graph on $n$ vertices, with $V=V(G)$ its vertex set and $E=E(G)$ its edge set. A vertex $x$ is a {\em neighbor} of $y$ if $xy$ is an edge of $E$. If $S$ is a subset of $V,$  then $G[S]$ is the {\em subgraph of} $G$ {\em induced by} $S$. A path $P$ between vertex $x$ and $y$ is called an {\em $(x,y)$-path} and the vertices of $P$ distinct from $x$ and $y$ are the {\em internal vertices} of $P$. A vertex $x\in V$ is a {\em cut vertex} if 
the removal of $x$ strictly increases the number of connected components of the graph.
A graph is {\em biconnected} if it is connected and does not contain a cut vertex. Given an integer $q,$  we use $[q]$ as a shortcut for  the set $\{1,\ldots,q\}$.

A \emph{layout} $\sigma$ of a graph $G=(V,E)$ is a total ordering of its vertices, in other words  $\sigma$ is a bijection from $V$ to $[n]$. 
For two vertices $x$ and $y,$  we write $x<_{\sigma} y$ if $\sigma(x)<\sigma(y)$. We define $\sigma_{<i}=\{x\in V\mid \sigma(x)<i\}$ (the sets $\sigma_{>i},$  $\sigma_{\leq i}$ and $\sigma_{\geq i}$ are similarly defined). If $S$ is a subset of $V,$  then $\sigma[S]$ is the layout of $G[S]$ such that for every $x,y\in S,$  $\sigma(x)<\sigma(y)$ if and only if $\sigma[S](x)<\sigma[S](y)$. Let $\sigma_1$ and $\sigma_2$ be two layouts on disjoint vertex sets $V_1$ and $ V_2$. Then $\sigma=\sigma_1\odot\sigma_2,$  the \emph{concatenation} of $\sigma_1$ and $\sigma_2,$  
is the layout on $V_1\cup V_2$ such that
for every $x_1\in V_1$ and every $x_2\in V_2,$  $\sigma(x_1)<\sigma(x_2),$  $\sigma[V_1]=\sigma_1$ and $\sigma[V_2]=\sigma_2$.

%---------------------------------------------------------------------------------------------------------------
\subsection{Series-parallel graphs}

A \emph{$2$-terminal graph} is a pair $\G{} =(G,(x,y))$ where $G=(V,E)$ is a graph and $(x,y)$ is a pair of distinguished vertices, hereafter called the \emph{terminals}. Consider two $2$-terminal graphs $(G_1,(x_1,y_1))$ and $(G_2,(x_2,y_2)),$  where $G_1=(V_1,E_1)$ and $G_2=(V_2,E_2)$. Then:

 %the \emph{series-composition}, denoted $(G_1,(s_1,t_1))\odot (G_2,(s_2,t_2)),$  and the \emph{parallel-composition}, denoted $(G_1,(s_1,t_1))\oplus (G_2,(s_2,t_2)),$  are defined as follows:
\begin{itemize}
\item  the \emph{series-composition}, denoted $(G_1,(x_1,y_1))\otimes (G_2,(x_2,y_2)),$  yields the $2$-terminal graph $(G,(x_1,y_2))$ with $G$ being the graph obtained by identifying the terminal $y_1$ with $x_2$;
\item the \emph{parallel-composition}, denoted $(G_1,(x_1,y_1))\oplus (G_2,(x_2,y_2)),$  yields the $2$-terminal graph $(G,(x_1,y_1))$ with $G$ being the graph obtained by identifying the terminal $x_1$ with $x_2$ and the terminal $y_1$ with $y_2$.
\end{itemize}

\begin{definition}
A $2$-terminal graph $(G,(x,y))$ is \emph{series-parallel} if either it is the single edge graph with $\{x,y\}$ as vertex set, or if it results from the series or the parallel composition of two $2$-terminal series-parallel graphs. A graph $G=(V,E)$ is a \emph{series-parallel graph} if for some pair of vertices $x,y\in V,$  $(G,(x,y))$ is a $2$-terminal series-parallel graph.
\end{definition}

Observe that from the definition above, we may generate multi-graphs. However, in this paper we only consider simple graphs. Our results extend easily to multi-graphs: observe that if $G^*$ is the graph obtained from $G$ after suppressing multiple edges to simple ones, then ${\ctw}(G)={\ctw}(G^*)$.
When it is clear from the context, $\G$ will denote the $2$-terminal series-parallel graph $(G,(x,y))$. Observe that a series-parallel graph $\G$ can be represented by a so-called \emph{\SPtree} $T(\G)$. The leaves of $T(\G)$ are labelled by the edges of $G$. Every internal node of $T(\G)$ is labelled by a composition operation ($\oplus$ or $\otimes$) and a pair of terminal vertices. For an internal node $t$ of $T(\G),$  we let $T_t$ denote the subtree of $T(\G)$ rooted at $t$ and $V_t$ the subset of vertices incident to an edge labelling a leaf of  $T_t$. Suppose that $t$ is labelled $(\otimes,(x_t,y_t)),$  then the node $t$ represents the $2$-terminal graph $(G[V_t],(x_t,y_t))$.

\begin{theorem}~\cite{Epp92}\label{th:SP-2terminal}
If a graph $G=(V,E)$ is a biconnected series-parallel graph and  $xy\in E$, then $(G,(x,y))$ is a $2$-terminal series-parallel graph.
\end{theorem}

Theorem~\ref{th:SP-2terminal} can be rephrased as follows: if $xy$ is an edge of a biconnected series-parallel graph, then $\G=(G,(x,y))$ is a $2$-terminal series-parallel graph such that  $\G=(G_1,(x,y))\oplus (G_2,(x,y))$ where $G_1=(\{x,y\},\{xy\})$ and $G_2=(V,E\setminus\{xy\})$.

\begin{theorem}~\cite{VTL82}\label{th:SP-recognition}
The problem of testing whether a graph $G$ is series-parallel can be solved in linear time. Moreover if $G$ is a biconnected series-parallel graph, then a \SPtree{} of $(G,(x,y)),$  where $xy$ is an edge, can be build in linear time.
\end{theorem}

%---------------------------------------------------------------------------------------------------------------
\subsection{Connected tree-decomposition and connected layouts}

A \emph{tree-decomposition} of a graph $G=(V,E)$ is pair $(T,\mathcal{F})$ where $T=(V_T,E_T)$ is a tree and $\mathcal{F}=\{X_t\subseteq V\mid t\in V_T\}$ such that
\begin{enumerate}
\item $\bigcup_{t\in V_T} X_t=V$;
\item for every edge $xy\in E,$  there exists a node $t\in V_T$ such that $\{x,y\}\subseteq X_t$;
\item for every vertex $x\in V,$  the set $\{t\in V_T\mid x\in X_t\}$ induces a connected subgraph of $T$.
\end{enumerate}
We refer to $V_T$ as the set of nodes of $T$ and the sets of $\mathcal{F}$ as the bags of $(T,\mathcal{F})$. The width of a tree-decomposition $(T,\mathcal{F})$ is $\mathsf{width}(T,\mathcal{F})=\max\{|X|-1\mid X\in \mathcal{F}\}$ and the treewidth of a graph $G$ is 
$$\mathsf{tw}(G)=\min\{\mathsf{width}(T,\mathcal{F})\mid (T,\mathcal{F}) \mbox{ is a tree-decomposition of } G\}$$

For two nodes $u$ and $v$ of $V_T,$  we define $P_{u,v}$ as the unique path between $u$ and $v$ in $T$ and 
 the set $V_{u,v}\subseteq V$ as $\bigcup_{t\in P_{uv}} X_t$.
A tree-decomposition $(T,\mathcal{F})$ is \emph{connected} if there exists a node $r\in V_T$ such that for every node $u\in V_T,$  the subgraph $G[V_{r,u}]$ is connected. Then the connected treewidth of a graph $G$ is 
$$\mathsf{ctw}(G)=\min\{\mathsf{width}(T,\mathcal{F})\mid (T,\mathcal{F}) \mbox{ is a connected tree-decomposition of } G\}$$

%A \emph{layout} $\sigma$ of a graph $G=(V,E)$ is a total ordering of its vertices, that $\sigma$ is a bijection from $V$ to $[n]$. For a vertex $x\in V,$  we set $\sigma_i=\sigma(x)$. For two vertices $x$ and $y,$  we write $x<_{\sigma} y$ if $\sigma(x)<\sigma(y)$. We may also simplify the notation by writing $\sigma_i=\sigma^{-1}(i)$. Finally we set $\sigma_{<i}=\{x\in V\mid \sigma(x)<i\}$ (the sets $\sigma_{>i},$  $\sigma_{\leq i}$ and $\sigma_{\geq i}$ are similarly defined). 
A layout $\sigma$ of $G$ is \emph{connected} if for every $i\in[n],$  the subgraph $G[\sigma_{\leq i}]$ is connected. We let $\mathcal{L}^c(G)$ be the set of connected layouts of a graph $G$.
%Let $\sigma$ be a layout of $G=(V,E)$. 
For every 
%$i\in [n],$  we define the \emph{supporting set} of $i$ as
$v\in V,$  we define the \emph{supporting set} of $v$ as
\begin{eqnarray*}
%S_{\sigma}(i) = \big\{x\in V(G)\mid \sigma(x)<i~\mbox{and there exists a $(x,\sigma_{i})$-path whose} \\
S_{\sigma}(v) = \big\{x\in V(G)\mid \sigma(x)<\sigma(v)~\mbox{and there exists a $(x,v)$-path $P$ such that} \\
 \mbox{\hfill every internal vertex $y$ of $P$ satisfies $\sigma(y)>\sigma(v)$}\}.
\end{eqnarray*}
We  set $\mathsf{cost}(G,\sigma)=\max \{|S_{\sigma}(v)|\mid v\in V\}$. The \emph{tree vertex separation number} of a graph is defined as 
$$\mathsf{vs}(G)=\min\{\mathsf{cost}(G,\sigma)\mid \sigma\in \mathcal{L}(G)\}.$$ 
When restricting to the set of connected  layouts, we obtain the \emph{connected tree vertex separation number} 
$$\mathsf{cvs}(G)= \min\{\mathsf{cost}(G,\sigma)\mid \sigma\in \mathcal{L}^{c}(G)\}.$$

\begin{theorem}[{\rm \cite{APT19}}] \label{th:equiv} 
For every graph $G=(V,E),$   we have $\ctw(G)=\cvs(G)$.
\end{theorem}

Notice that if in the above definitions we drop the connectivity demand from the considered layouts, we have that $\tw(G)=\mathsf{vs}(G)$ providing an alternative layout-definition of  the parameter of treewidth, as observed in~\cite{Arn85,DKT97}.
It is known that deciding whether $\tw(G)\leq k$ is an {\sf NP}-complete problem \cite{ACP87}. An easy reduction of treewidth to connected treewidth is the following: consider a graph $G,$  add a vertex $v_{\rm new},$  and make $v_{\rm new}$ adjacent to all the vertices of $G$. We call the new graph $G^+$.
It follows, as a direct consequence of the layout definitions, that 
$\ctw(G^+) = \tw(G)+1$. We conclude to the following.

\begin{theorem}
 \label{th_eqdcoml}
The problem of deciding, given a graph $G$ and an integer $k,$  whether $\ctw(G)\leq k$ is {\sf NP}-complete.
\end{theorem}

%---------------------------------------------------------------------------------------------------------------
\subsection{Rooted graphs and extended graphs}

\paragraph{Rooted graphs.} A \emph{rooted graph} is a pair $(G,\R)$ where $G=(V,E)$ is a graph and $\R\subseteq V$ is a subset of vertices, hereafter called \emph{roots}. 
%Observe that every graph $G$ can be considered as the rooted graph $(G,\emptyset)$. 
The definition of a rooted graph naturally extends to two-terminal graphs. If $\G=(G,(x,y))$ is a series-parallel graph and $\R\subseteq V$ a set of roots, then the corresponding rooted two-terminal graph will be denoted by $(\G,\R)$. Observe that the set of roots $\R$ may be different from the terminal pair $(x,y)$.

A rooted graph $(G,\R)$ is connected if and only if every connected component of $G$ contains a root from $\R$. A \emph{rooted layout} of $(G,\R)$ is a layout $\sigma$ of $G$ such that $\sigma_{\leq |\R|}=\R$. Based on this, the notion of connected layout naturally extends to rooted graphs and rooted layouts as follows. A rooted layout $\sigma$ of $(G,\R)$ is connected if for every $i,$  $|\R|<i\leq n,$  $(G[\sigma_{\leq i}],\R)$ is a connected rooted graph.

\paragraph{Extended graphs.} Let $\G=(G,(x,y))$ be a $2$-terminal graph where $G=(V,E)$. Suppose that $F\subseteq {V\choose 2 }\setminus E(G),$  i.e., $F$  is a subset of non-edges of $G$. We define the \emph{extended} graph $\G^{+F}$ as the $2$-terminal graph $(G^{+F},(x,y))$ where $G^{+F}=(V,E\cup F)$. The
edges in $E$ are called \emph{solid} edges, while the edges of $F$ are called \emph{fictive} edges. Hereafter the $2$-terminal graph $\G$ (and the graph $G$ respectively) is called the \emph{solid graph} of $\G^{+F}$ (and $G^{+F}$ respectively). 

We define the connected components of the extended graph $\G^{+F}$ as the connected components of its solid graph $\G$. In particular, $\G^{+F}$ is connected if and only if $\G$ is connected. Thereby, we say a vertex is isolated in $\G^{+F}$ if it is isolated in $\G$. Likewise, we say that two vertices are adjacent in $\G^{+F}$ if they are in $\G$. Therefore the neighbourhood $N(v)$ of a vertex $v$ in $\G^{+F}$ is its neighbourhood in $\G$, while its extended neighbourhood $N^{+F}(v)$ also comprises every vertex $u$ such that $uv\in F$ (we say that $u$ is an \emph{extended-neighbour} of $v$). The purpose of introducing fictive edges is {\sl not} to augment the connectivity of the solid graph but to keep track of cumulative cost in the recursive calls of the dynamic programming algorithm while computing the connected treewidth of a series-parallel graph. 

This connectivity definition of an extended graph also transfers to (rooted) layouts. More precisely, if $\R\subseteq V$ is a set of roots and $\G^{+F}$ an extended graph, then $(\G^{+F},\R)$ is a \emph{rooted extended graph}.  A layout $\sigma$ of $(\G^{+F},\R)$ is connected if and only if it is a connected layout of $(\G,\R)$. Observe that if $\G^{+F}$ is not connected, then a connected layout of $\G^{+F}$ exists if and only if every connected component of $\G^{+F}$ contains a root from $\R$. 

An \emph{extended path} of $\G^{+F}$ is a path that may contain a fictive edge of $F,$  while a \emph{solid path} in $\G^{+F}$ is a path of $\G,$  that is, a path that only contains solid edges. In a layout $\sigma$ of the extended graph $\G^{+F}$, the \emph{extended supporting set} of  vertex $v$ is:
\begin{eqnarray*}
S^{(e)}_{\sigma}(v) = \big\{x\in V\mid \sigma(x)<\sigma(v)~\mbox{and there exists an extended $(x,v)$-path $P$ such that }\\
 \mbox{\hfill every internal vertex $y$ of $P$ belongs to $\sigma_{>i}$}\}.
\end{eqnarray*}

\noindent
The definitions of the extended cost $\ecost(\G^{+F},\sigma)$ and of the extended connected tree vertex separation number $\ecvs(\G^{+F})$ follow accordingly: 
$$\mathsf{ecvs}(\G^{+F})= \min\{\mathsf{ecost}(\G^{+F},\sigma)\mid \sigma\in \mathcal{L}^{c}(\G^{+F})\},$$
 where $\ecost(\G^{+F},\sigma)=\max\{|S^{(e)}_{\sigma}(v)|\mid v\in V\}$.

%---------------------------------------------------------------------------------------------------------------
%---------------------------------------------------------------------------------------------------------------
\section{A dynamic programming algorithm}

\label{adpal}

Before describing the dynamic programming algorithm to compute the connected treewidth of treewidth at most $2$ graphs, we examine the case of biconnected series-parallel graphs. More precisely, we show how to derive the connected treewidth of a series-parallel graph depending on whether it results from a series or from a parallel composition.  To handle these recursion steps, we have to manipulate extended series-parallel graphs.

\subsection{Biconnected series-parallel graphs}

In this section, we let $\G=(G,(x,y))$ be a series-parallel graph such that $G=(V,E)$. We suppose that $G$ results from the series or the parallel composition of $G_1=(V_1,E_1)$ and $G_2=(V_2,E_2)$. When dealing with the $2$-terminal graphs $\G_1$ and $\G_2,$  the terminal pairs will be clear from the context.

\subsubsection{Parallel composition}

\begin{lemma} \label{lem:parallel-no-fictive}
Let $(\G^{+\emptyset},\R)$ be a rooted extended graph such that $\R=\{x,y\}$ and 
$\G=\G_1\oplus \G_2$ with $\G_1=(G_1,(x,y))$ and $\G_2=(G_2,(x,y))$  (see Figure~\ref{fig:parallel}).
Then

$$\ecvs(\G^{+\emptyset},\R)=\max \left\{\ecvs(\G_1^{+\emptyset},\R), \ecvs(\G_2^{+\emptyset},\R) \right\}.$$
\end{lemma}
\begin{proof}
Let $\sigma^*\in\mathcal{L}^c(\G^{+\emptyset},\R)$ be a connected layout of minimum cost. From $\sigma^*,$  we define a layout $\sigma$ of $(\G^{+\emptyset},\R)$ as follows: 
$\sigma=\sigma^*[V_1]\odot \sigma^*[V_2\setminus\{x,y\}]$ (see Figure~\ref{fig:layout-rearrangement-parallel}).

\begin{figure}[h]
\centering
\begin{tikzpicture}[scale=0.9]
   \tikzstyle{vertex}=[fill,circle,minimum size=0.2cm,inner sep=0pt]
   \tikzstyle{vertexRoot}=[fill,red,rectangle,minimum size=0.2cm,inner sep=0pt]
   \tikzstyle{vertex2}=[draw,circle,minimum size=0.2cm,inner sep=0pt]
   \tikzstyle{diamondb}=[fill=red!20,draw,diamond,minimum size=0.25cm,inner sep=0pt]
   \tikzstyle{diamondw}=[fill=blue!20,draw,diamond,minimum size=0.25cm,inner sep=0pt]

%\sigma*
\draw[thin,->,arrows={-latex'}] (-5,0) -- (5.5,0);
\node[anchor=east] at (-5,0) {$\sigma^*$};

\foreach \i/\posnd in {1/-4.5,2/-4}
{\node[vertexRoot] (nd\i) at (\posnd,0){};}

\node[anchor=north] at (-4.5,0.6) {$x$};
\node[anchor=north] at (-4,0.6) {$y$};

\foreach \i/\posnd in {3/-3.5,4/-3,5/-2.5,6/-2,9/-0.5,10/0,15/2.5,16/3,18/4}
{\node[diamondb] (nd\i) at (\posnd,0){};}
\foreach \i/\posnd in {7/-1.5,8/-1,11/0.5,12/1,13/1.5,14/2,17/3.5,19/4.5,20/5}
{\node[diamondw] (nd\i) at (\posnd,0){};}

%\sigma
\draw[thin,->,arrows={-latex'}] (-5,-1.5) -- (5.5,-1.5);
\node[anchor=east] at (-5,-1.5) {$\sigma$};

\foreach \i/\posnd in {1/-4.5,2/-4}
{\node[vertexRoot] (nds\i) at (\posnd,-1.5){};}

\node[anchor=north] at (-4.5,-1.7) {$x$};
\node[anchor=north] at (-4,-1.7) {$y$};

\foreach \i/\posnd in {3/-3.5,4/-3,5/-2.5,6/-2,9/-1.5,10/-1,15/-0.5,16/0,18/0.5}
{\node[diamondb] (nds\i) at (\posnd,-1.5){};}
\foreach \i/\posnd in {7/1,8/1.5,11/2, 12/2.5, 13/3,14/3.5,17/4,19/4.5,20/5}
{\node[diamondw] (nds\i) at (\posnd,-1.5){};}

\foreach \i in {1,2,3,4,5,6,7,8,9,10,11,12,13,14,15,16,17,18,19,20}
{\draw[gray!60,thick,dotted] (nd\i) -- (nds\i);
}
\end{tikzpicture}
\caption{Rearranging a layout of minimum cost of an extended graph $\G=\G_1\oplus\G_2$ without a fictive edge. Red vertices belongs to $V_1\setminus\{x,y\}$ and blue vertices belong to $V_2\setminus\{x,y\}$. \label{fig:layout-rearrangement-parallel}}
\end{figure}
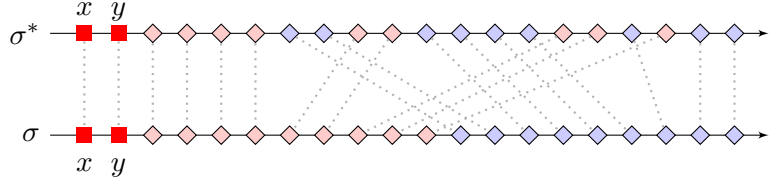

Observe that as $\{x,y\}$ separates $V_1$ from $V_2$ and as $\sigma^*\in\mathcal{L}^c(\G^{+\emptyset},\R)$, it follows that $\sigma\in\mathcal{L}^c(\G^{+\emptyset},\R),$  $\sigma_1=\sigma[V_1]\in\mathcal{L}^c(\G_1^{+\emptyset},\R)$ and $\sigma_2=\sigma[V_2]\in\mathcal{L}^c(\G_2^{+\emptyset},\R)$.
Moreover, $\{x,y\}$ separating $V_1$ from $V_2$ implies that for every vertex $v_1\in V_1\setminus\{x,y\},$  we have $S^{(e)}_{\sigma^*}(v_1)=S^{(e)}_{\sigma_1}(v_1)\subseteq V_1$ and that for every vertex $v_2\in V_2\setminus\{x,y\},$  we have $S^{(e)}_{\sigma^*}(v_2)=S^{(e)}_{\sigma_2}(v_2)\subseteq V_2$. It follows that for every vertex $v\in V\setminus\{x,y\},$  $S^{(e)}_{\sigma^*}(v)=S^{(e)}_{\sigma}(v),$  implying that $\ecost(\G^{+\emptyset},\sigma)=\ecost(\G^{+\emptyset},\sigma^*)$. 

%The facts that $\{x,y\}$ separates $V_1$ from $V_2$ and that $\sigma^*\in\mathcal{L}^c(G^{+\emptyset},\R)$ imply that $\sigma\in\mathcal{L}^c(G^{+\emptyset},\R),$  $\sigma_1=\sigma[V1]\in\mathcal{L}^c(G_1^{+\emptyset},\R)$ and $\sigma_2=\sigma[V2]\in\mathcal{L}^c(G_2^{+\emptyset},\R)$. 

%It follows that (1) $\ecvs(G_1^{+\emptyset},\R)\leq \ecost((G_1^{+\emptyset},\R),\sigma_1)$ and (2) $\ecvs(G_2^{+\emptyset},\R)\leq \ecost((G_2^{+\emptyset},\R),\sigma_2)$. 

\begin{figure}[htbh]
\centering
\begin{tikzpicture}[scale=0.9]
   \tikzstyle{vertex}=[fill,circle,minimum size=0.2cm,inner sep=0pt]
   \tikzstyle{vertexRoot}=[fill,red,rectangle,minimum size=0.15cm,inner sep=0pt]

%---- G------
\filldraw[fill=red!20]
  (0,0) .. controls (-1.5,1) .. (0,2) .. controls (-0.5,1) .. (0,0)
  -- cycle ;
\node[] at (-0.7,1) {$G_1$};

\filldraw[fill=blue!20]
  (0,0) .. controls (1.5,1) .. (0,2) .. controls (0.5,1) .. (0,0)
  -- cycle ;
\node[] at (0.7,1) {$G_2$};

\node[vertexRoot]  (y) at (0,2){};
\node[vertexRoot] (x) at (0,0){};

\node[left] at (y) {$y$};
\node[below] at (x) {$x$};

\draw[black,very thick,->] (2.25,1) -- (2.75,1) ;

%---- G1------
\filldraw[fill=red!20]
  (5,0) .. controls (3.5,1) .. (5,2) .. controls (4.5,1) .. (5,0)
  -- cycle ;
\node[] at (4.3,1) {$G_1$};

\node[vertexRoot]  (yy) at (5,2){};
\node[vertexRoot] (xx) at (5,0){};

\node[left] at (yy) {$y$};
\node[below] at (xx) {$x$};

\draw[black,very thick] (6.3,1) -- (6.7,1) ;
\draw[black,very thick] (6.5,1.2) -- (6.5,0.8) ;

%---- G2------

\filldraw[fill=blue!20]
  (8,0) .. controls (9.5,1) .. (8,2) .. controls (8.5,1) .. (8,0)
  -- cycle ;
\node[] at (8.7,1) {$G_2$};

\node[vertexRoot]  (yyy) at (8,2){};
\node[vertexRoot] (xxx) at (8,0){};

\node[left] at (yyy) {$y$};
\node[below] at (xxx) {$x$};

%\node[vertexRoot] (y) at (-2,-1){};

%\filldraw[fill=red!20]
%%(x) -- (y) --
% (x) to[bend left=30] (y) 
% to[bend left=-100] (x)
% -- cycle ;

%\draw (0,1) arc (225:135:2) ;
%\filldraw[draw=black,fill=red!20]
%   plot [smooth,domain=0:1] (\x,{sqrt(\x)})
%-- plot [smooth,domain=1:0] (\x,\x^2)
%-- cycle;

%\filldraw[black] 
%(0,0) circle [radius=2pt]
%%(1,1) circle [radius=2pt] 
%%(2,1) circle [radius=2pt] 
%(0,2) circle [radius=2pt];

%\filldraw[fill=red!20]
 % (0,0) .. controls (-1.5,1) .. (0,2) .. controls (-0.5,1) .. (0,0)
 % -- cycle ;

%\newcommand{\A}{(0,0) ++(35:2) circle (1)}
%\newcommand{\B}{(0,0) ++(10:2) circle (1)}
%
%\tikz\draw\A\B;
%
%\begin{scope}
%    \clip \B;
%    \fill[red] \A;
%\end{scope}

\end{tikzpicture}
\caption{Parallel decomposition of an extended graph without fictive edges. If $\G=\G_1\oplus\G_2$, then we have $\ecvs(\G^{+\emptyset},\R)=\max \left\{\ecvs(\G_1^{+\emptyset},\R), \ecvs(\G_2^{+\emptyset},\R) \right\}.$
\label{fig:parallel}}
\end{figure}
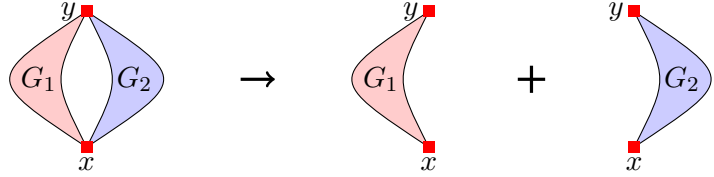

To conclude, we observe that if $\tau_1\in \mathcal{L}^c(\G_1^{+\emptyset},\R)$ and $\tau _2\in\mathcal{L}^c(\G_2^{+\emptyset},\R),$  then $\tau=\tau_1\odot \tau_2[V_2\setminus\{x,y\}]$ belongs to $\mathcal{L}^c(\G^{+\emptyset},\R)$ and that $\ecost(\G^{+\emptyset},\tau)=\max\{\ecost(\G_1^{+\emptyset},\tau_1),\ecost(\G_2^{+\emptyset},\tau_2)\}$. So the optimality of $\sigma^*$ and $\sigma$ imply that $\ecvs(\G^{+\emptyset},\R)=\max \left\{\ecvs(\G_1^{+\emptyset},\R), \ecvs(\G_2^{+\emptyset},\R) \right\}$.
% the inequalities (1) and (2) are indeed equalities, proving the statement.
\end{proof}

\begin{lemma} \label{lem_parallel_fictive}
Let $(\G^{+F},\R)$ be a rooted extended graph such that $\R=\{x,r_1\dots, r_k\}$ (with $k>0$), $r_1,\dots r_k$ are isolated vertices (in $\G$), $F=\{yr_i\mid i\in[k]\}$ and 
%$(G[V\setminus\{r_1,\dots r_k\}],(x,y))=(G_1,(x,y))\oplus (G_2,(x,y))$. 
$(G[V\setminus\{r_1,\dots, r_k\}],(x,y))=\G_1\oplus \G_2$ with $\G_1=(G_1,(x,y))$ and $\G_2=(G_2,(x,y))$
(see Figure~\ref{fig:parallel-fictive1}).
Then

$$\ecvs(\G^{+F},\R)
=\min \left\{
\begin{array}{c} 
\max \left\{\ecvs(\G[V_1\cup\R]^{+F\cup\{xy\}},\R), \ecvs(\G_2^{+\emptyset},\{x,y\}) \right\}\\ 
\\
\max \left\{\ecvs(\G[V_2\cup\R]^{+F\cup\{xy\}},\R), \ecvs(\G_1^{+\emptyset},\{x,y\}) \right\} \\
\end{array}\right\}.
$$
%where $V_1$ and $V_2$ respectively denote the vertices of $G_1$ and $G_2$.

\end{lemma}

%
%\begin{figure}[h]
%\centerline{\includegraphics[width=12cm]{parallel-1-fictive}}
%\caption{}
%\end{figure}
\begin{proof}
We observe that $\G^{+F}$ contains $k$ isolated vertices, the root vertices $r_1,\dots, r_k,$  and a connected component resulting from a parallel composition. Let $\sigma^*\in\mathcal{L}^c(\G^{+F},\R)$ be a connected layout of minimum cost. Consider the neighbor $v$ of $y$ so that $\sigma^*(v)$ is minimized. By the connectivity of $\sigma^*,$  $\sigma^*(v)<\sigma^*(y)$. Suppose without loss of generality that $v\in V_1$. The case $v\in V_2$ is symmetric.
{Observe that $v$ can be the vertex $x,$  in which case it can arbitrarily be considered as a vertex of $V_1$ and $V_2$).}  
%From $\sigma^*,$  we define a layout $\sigma$  of $(\G^{+F},\R)$ as follows: $\sigma=\langle r_1,\dots, r_k\rangle\odot \sigma^*[V_1]\odot \sigma^*[V_2\setminus\{x,y\}]$.
% (see Figure~\ref{fig:layout-rearrangement-parallel}).
\xchange{
From $\sigma^*,$  we define the layout $\sigma_1=\sigma[V_1\cup \R]$  of $(\G[V_1\cup\R]^{+F\cup\{xy\}},\R)$ and the layout $\sigma_2=\sigma[V_2]$ of $(\G_2,\{x,y\})$,  where $\sigma=\langle r_1,\dots, r_k\rangle\odot \sigma^*[V_1]\odot \sigma^*[V_2\setminus\{x,y\}]$ is a layout of $(\G^{+F},\R)$  (see Figure~\ref{fig:layout-rearrangement-parallel}).
}

%$\sigma_1=\sigma[V_1\cup \R]\in\mathcal{L}^c(\G[V_1\cup\R]^{+F\cup\{xy\}},\R)$. 

%$\sigma_2=\sigma[V_2]\in\mathcal{L}^c(\G_2,\{x,y\})$.

\begin{claim} \label{cl:parallel-fictive-1}
$\sigma_1\in\mathcal{L}^c(\G[V_1\cup\R]^{+F\cup\{xy\}},\R)$,
$\sigma_2\in\mathcal{L}^c(\G_2,\{x,y\})$ and 
$\sigma\in\mathcal{L}^c(\G^{+F},\R)$. 
\end{claim}
\begin{proofclaim}
Let $v_1$ be a vertex of $V_1$ distinct from $x$ and $y$. Observe that every neighbor of $v_1$ belongs to $V_1$. By the assumption above, we know that $y$ has a neighbor $v\in V_1$ prior to it in $\sigma^*$. As the relative ordering between vertices of $V_1$ is left unchanged, every vertex of $V_1\setminus\{x\}$ has a neighbor prior to it in $\sigma$ as well. 
\xchange{This implies that $\sigma_1\in\mathcal{L}^c(\G[V_1\cup\R]^{+F\cup\{xy\}},\R)$.}
Consider now a vertex $v_2\in V_2$ distinct from $x$ and $y$. As in the previous case,  every neighbor of $v_2$ belongs to $V_2$. As the relative ordering between vertices of $V_2$ has only been modified by moving $y$ ahead, $v_2$ has a neighbor prior to it in $\sigma$. 
\xchange{This implies that $\sigma_2\in\mathcal{L}^c(\G_2,\{x,y\})$ and that $\sigma\in\mathcal{L}^c(\G^{+F},\R)$.}
%It follows that $\sigma\in\mathcal{L}^c(\G^{+F},\R)$.
\end{proofclaim}

\begin{claim} \label{cl:parallel-fictive-2}
{$\ecost((\G^{+F},\R),\sigma)= \ecost((\G^{+F},\R),\sigma^*)$.}
\end{claim}

\begin{proofclaim}
We first consider a vertex $v_1\in V_1$. By construction, we have $\sigma^*[V_1]=\sigma[V_1]$ and for every vertex $v_2\in V_2\setminus\{x,y\},$  $\sigma(v_1)\leqslant \sigma(v_2)$. It follows that $S^{(e)}_{\sigma}(v_1)\subseteq V_1\cup\R$. Suppose that $\sigma(v_1)>\sigma(y)$. As $\{x,y\}$ separates $v_1$ from the vertices of $V_2\cup\{r_1,\dots, r_k\},$  we have $S^{(e)}_{\sigma}(v_1)=S^{(e)}_{\sigma^*}(v_1)$. Suppose that $\sigma(v_1)\leq\sigma(y)$. We observe that if for $i\in[k],$  $r_i\in S^{(e)}_{\sigma}(v_1),$  then there exists a $(v_1,y)$-path $P$ in $G_1$ such that every vertex $v\in P$ satisfies $\sigma(v_1)\leq \sigma(v)$. Similarly, if $\xchange{u}\in V_1$ belongs to $S^{(e)}_{\sigma}(v_1),$  then there exists a $(v_1,\xchange{u})$-path $P'$ in $G_1$ such that every vertex $v'\in P'$ distinct from \xchange{$u$} satisfies $\sigma(v_1)\leq \sigma(v')$. The existence of the paths $P$ and $P'$ is a consequence of the fact that $\{x,y\}$ separates the vertices of $V_1\setminus\{x,y\}$ from the rest of the graph. As $\sigma^*[V_1\cup\R]=\sigma[V_1\cup\R],$  we have $\sigma^*(v_1)\leq\sigma^*(v)$ and $\sigma^*(v_1)\leq\sigma^*(v')$. It follows that $r_i,u\in S^{(e)}_{\sigma^*}(v_1)$, implying \xchange{in turn} that $|S^{(e)}_{\sigma}(v_1)|\leq |S^{(e)}_{\sigma^*}(v_1)|$.

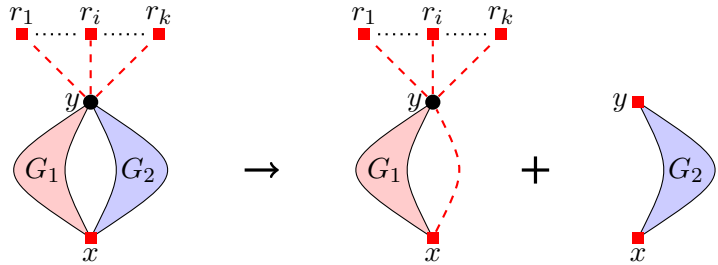
\begin{figure}[h]
\centering
\begin{tikzpicture}[scale=0.9]
   \tikzstyle{vertex}=[fill,circle,minimum size=0.2cm,inner sep=0pt]
   \tikzstyle{vertexRoot}=[fill,red,rectangle,minimum size=0.15cm,inner sep=0pt]

%---- G------
\filldraw[fill=red!20]
  (0,0) .. controls (-1.5,1) .. (0,2) .. controls (-0.5,1) .. (0,0)
  -- cycle ;
\node[] at (-0.7,1) {$G_1$};

\filldraw[fill=blue!20]
  (0,0) .. controls (1.5,1) .. (0,2) .. controls (0.5,1) .. (0,0)
  -- cycle ;
\node[] at (0.7,1) {$G_2$};

\node[vertex]  (y) at (0,2){};
\node[vertexRoot] (x) at (0,0){};

\foreach \k in {0,1,2}{
	\node[vertexRoot] (r\k) at (\k-1,3){};
	\draw[red,thick,dashed] (y) -- (r\k);
}

\node[above] at (r0) {$r_1$};
\node[above] at (r1) {$r_i$};
\node[above] at (r2) {$r_k$};
\node[left] at (y) {$y$};
\node[below] at (x) {$x$};

\draw[black,thick,dotted] (-0.8,3) -- (-0.2,3) ;
\draw[black,thick,dotted] (0.2,3) -- (0.8,3) ;

\draw[black,very thick,->] (2.25,1) -- (2.75,1) ;

%---- G1------
\filldraw[fill=red!20]
  (5,0) .. controls (3.5,1) .. (5,2) .. controls (4.5,1) .. (5,0)
  -- cycle ;
\node[] at (4.3,1) {$G_1$};

\node[vertex]  (yy) at (5,2){};
\node[vertexRoot] (xx) at (5,0){};

\foreach \k in {0,1,2}{
	\node[vertexRoot] (rr\k) at (\k+4,3){};
	\draw[red,thick,dashed] (yy) -- (rr\k);
}

\draw[red,thick,dashed] (xx) .. controls (5.5,1) .. (yy) ;

\node[above] at (rr0) {$r_1$};
\node[above] at (rr1) {$r_i$};
\node[above] at (rr2) {$r_k$};
\node[left] at (yy) {$y$};
\node[below] at (xx) {$x$};

\draw[black,thick,dotted] (4.2,3) -- (4.8,3) ;
\draw[black,thick,dotted] (5.2,3) -- (5.8,3) ;

\draw[black,very thick] (6.3,1) -- (6.7,1) ;
\draw[black,very thick] (6.5,1.2) -- (6.5,0.8) ;

%---- G2------

\filldraw[fill=blue!20]
  (8,0) .. controls (9.5,1) .. (8,2) .. controls (8.5,1) .. (8,0)
  -- cycle ;
\node[] at (8.7,1) {$G_2$};

\node[vertexRoot]  (yyy) at (8,2){};
\node[vertexRoot] (xxx) at (8,0){};

\node[left] at (yyy) {$y$};
\node[below] at (xxx) {$x$};

%\node[vertexRoot] (y) at (-2,-1){};

%\filldraw[fill=red!20]
%%(x) -- (y) --
% (x) to[bend left=30] (y) 
% to[bend left=-100] (x)
% -- cycle ;

%\draw (0,1) arc (225:135:2) ;
%\filldraw[draw=black,fill=red!20]
%   plot [smooth,domain=0:1] (\x,{sqrt(\x)})
%-- plot [smooth,domain=1:0] (\x,\x^2)
%-- cycle;

%\filldraw[black] 
%(0,0) circle [radius=2pt]
%%(1,1) circle [radius=2pt] 
%%(2,1) circle [radius=2pt] 
%(0,2) circle [radius=2pt];

%\filldraw[fill=red!20]
 % (0,0) .. controls (-1.5,1) .. (0,2) .. controls (-0.5,1) .. (0,0)
 % -- cycle ;

%\newcommand{\A}{(0,0) ++(35:2) circle (1)}
%\newcommand{\B}{(0,0) ++(10:2) circle (1)}
%
%\tikz\draw\A\B;
%
%\begin{scope}
%    \clip \B;
%    \fill[red] \A;
%\end{scope}

\end{tikzpicture}
\caption{Decomposition of an extended graph with several isolated roots and resulting from a parallel composition.
\label{fig:parallel-fictive1}}
\end{figure}

Let $v_2$ be a vertex of $V_2\setminus\{x,y\}$. Observe that, as $\{x,y\}$ separates the vertices of $V_2\setminus\{x,y\}$ from the vertices of $\R\cup V_1\setminus\{x,y\}$ and as $\sigma(x)<\sigma(y)<\sigma(v_2),$  we have $S^{(e)}_{\sigma}(v_2)\subseteq V_2$. As $\sigma[V_2\setminus\{y\}]=\sigma^*[V_2\setminus\{y\}],$  every vertex $u\in V_2\setminus\{y\}$ that belongs to $S^{(e)}_{\sigma}(v_2)$ also belongs to $S^{(e)}_{\sigma^*}(v_2)$. %Observe that as $y$ is moved before the vertices of $V_2\setminus\{x,y\},$  none of the root vertices $r_1,\dots r_k$ belong to $S^{(e)}_{\sigma}(v_2)$. ,
Suppose that $y\in S^{(e)}_{\sigma}(v_2)$. Then there exists a $(y,v_2)$-path $P$ such that every internal vertex $u$ of $P$ satisfies $\sigma(v_2)<\sigma(u)$. As $\sigma[V_2\setminus\{y\}]=\sigma^*[V_2\setminus\{y\}],$  we also have $\sigma^*(v_2)<\sigma^*(v)$. Let us distinguish two cases:
\begin{itemize}
\item If $\sigma^*(y)<\sigma^*(v_2),$  then $y\in S^{(e)}_{\sigma}(v_2)$ implying that $S^{(e)}_{\sigma}(v_2)\subseteq S^{(e)}_{\sigma^*}(v_2)$. 
\item Otherwise, $\sigma^*(y)>\sigma^*(v_2)$ and then $y\notin S^{(e)}_{\sigma^*}(v_2)$. But in that case, let us recall that by assumption the first neighbor $v$ of $y$ in $\sigma^*$ belongs to $V_1$. It follows that $v\in S^{(e)}_{\sigma^*}(v_2)$. As we argue that $S^{(e)}_{\sigma}(v_2)\subseteq V_2,$  $v\in S^{(e)}_{\sigma^*}(v_2)\setminus S^{(e)}_{\sigma}(v_2)$. So $v$ is a replacement vertex for $y$ in $S^{(e)}_{\sigma^*}(v_2),$  implying that $|S^{(e)}_{\sigma}(v_2)|\leq |S^{(e)}_{\sigma^*}(v_2)|$.
\end{itemize}

So we proved that $\ecost((\G^{+F},\R),\sigma)\leq \ecost((\G^{+F},\R),\sigma^*)$. As by Claim~\ref{cl:parallel-fictive-1}, $\sigma\in\mathcal{L}^c(\G^{+F},\R),$  the optimality of $\sigma^*$ implies that $\ecost((\G^{+F},\R),\sigma)= \ecost((\G^{+F},\R),\sigma^*)$.
\end{proofclaim}

Let us now conclude the proof. 
Claim~\ref{cl:parallel-fictive-1}, Claim~\ref{cl:parallel-fictive-2} and the optimality of $\sigma^*$ imply that $\sigma$ is a connected layout of $\G^{+F}$ of minimum cost.
%From the proof of Claim~\ref{cl:parallel-fictive-1}, we deduce that $\sigma_1=\sigma[V_1\cup \R]\in\mathcal{L}^c(\G[V_1\cup\R]^{+F\cup\{xy\}},\R)$. 
\xchange{By Claim~\ref{cl:parallel-fictive-1}, we have $\sigma_1=\sigma[V_1\cup \R]\in\mathcal{L}^c(\G[V_1\cup\R]^{+F\cup\{xy\}},\R)$.}
{Observe that in the extended graph $\G[V_1\cup\R]^{+F\cup\{xy\}},$  the edge $xy$ simulates every $(x,y)$-path of $G$ whose internal vertices belong to $V_2$.}
It follows that for every vertex $v_1\in V_1\cup\R,$  $S^{(e)}_{\sigma}(v_1)=S^{(e)}_{\sigma_1}(v_1)$.
%Likewise, from the proof of Claim~\ref{cl:parallel-fictive-1}, we deduce that $\sigma_2=\sigma[V_2]\in\mathcal{L}^c(\G_2,\{x,y\})$. 
\xchange{Likewise, by Claim~\ref{cl:parallel-fictive-1}, we know that $\sigma_2=\sigma[V_2]\in\mathcal{L}^c(\G_2,\{x,y\})$.}
 As $\{x,y\}$ separates the vertices of $V_2$ from the other vertices, for every $v_2\in V_2\setminus\{x,y\},$  we have $S^{(e)}_{\sigma}(v_2)=S^{(e)}_{\sigma_2}(v_2)$.
It follows that $\ecvs(\G^{+F},\R)=\max \left\{\ecvs(\G[V_1\cup\R]^{+F\cup\{xy\}},\R), \ecvs(\G_2^{+\emptyset},\{x,y\}) \right\}$. (see Figure~\ref{fig:parallel-fictive1})
\end{proof}

The next lemma shows that, in the case of a parallel composition, a fictive edge between the two terminal vertices is irrelevant. This is because the graph is already biconnected.

\begin{lemma} \label{lem_parallel_fictivebis}
Let {$(\G^{+F\cup\{xy\}},\R)$} be a rooted extended graph such that $\R=\{x,r_1\dots, r_k\}$ (with $k>0$), $r_1,\dots r_k$ are isolated vertices (in $\G$), $F=\{yr_i\mid i\in[k]\}$ and $(G[V\setminus\{r_1,\dots, r_k\}],(x,y))=\G_1\oplus \G_2$ with $\G_1=(G_1,(x,y))$ and $\G_2=(G_2,(x,y))$. Then
$$\ecvs(\G^{+F\cup\{xy\}},\R)=\ecvs(\G^{+F},\R).$$
\end{lemma}

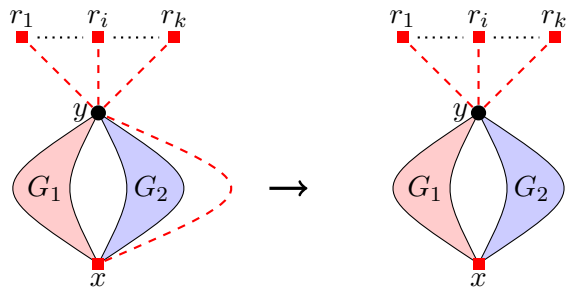
\begin{figure}[h]
\centering
\begin{tikzpicture}[scale=1.0]
   \tikzstyle{vertex}=[fill,circle,minimum size=0.2cm,inner sep=0pt]
   \tikzstyle{vertexRoot}=[fill,red,rectangle,minimum size=0.15cm,inner sep=0pt]

%---- G------
\filldraw[fill=red!20]
  (0,0) .. controls (-1.5,1) .. (0,2) .. controls (-0.5,1) .. (0,0)
  -- cycle ;
\node[] at (-0.7,1) {$G_1$};

\filldraw[fill=blue!20]
  (0,0) .. controls (1.5,1) .. (0,2) .. controls (0.5,1) .. (0,0)
  -- cycle ;
\node[] at (0.7,1) {$G_2$};

\node[vertex]  (y) at (0,2){};
\node[vertexRoot] (x) at (0,0){};

\draw[red,thick,dashed] (x) .. controls (2.3,1) .. (y) ;

\foreach \k in {0,1,2}{
	\node[vertexRoot] (r\k) at (\k-1,3){};
	\draw[red,thick,dashed] (y) -- (r\k);
}

\node[above] at (r0) {$r_1$};
\node[above] at (r1) {$r_i$};
\node[above] at (r2) {$r_k$};
\node[left] at (y) {$y$};
\node[below] at (x) {$x$};

\draw[black,thick,dotted] (-0.8,3) -- (-0.2,3) ;
\draw[black,thick,dotted] (0.2,3) -- (0.8,3) ;

\draw[black,very thick,->] (2.25,1) -- (2.75,1) ;

%---- G'------
\filldraw[fill=red!20]
  (5,0) .. controls (3.5,1) .. (5,2) .. controls (4.5,1) .. (5,0)
  -- cycle ;
\node[] at (4.3,1) {$G_1$};

\filldraw[fill=blue!20]
  (5,0) .. controls (6.5,1) .. (5,2) .. controls (5.5,1) .. (5,0)
  -- cycle ;
\node[] at (5.7,1) {$G_2$};

\node[vertex]  (yy) at (5,2){};
\node[vertexRoot] (xx) at (5,0){};

\foreach \k in {0,1,2}{
	\node[vertexRoot] (rr\k) at (\k+4,3){};
	\draw[red,thick,dashed] (yy) -- (rr\k);
}

%\draw[red,thick,dashed] (xx) .. controls (5.5,1) .. (yy) ;

\node[above] at (rr0) {$r_1$};
\node[above] at (rr1) {$r_i$};
\node[above] at (rr2) {$r_k$};
\node[left] at (yy) {$y$};
\node[below] at (xx) {$x$};

\draw[black,thick,dotted] (4.2,3) -- (4.8,3) ;
\draw[black,thick,dotted] (5.2,3) -- (5.8,3) ;

\end{tikzpicture}
\caption{The fictive edge $xy$ is irrelevant: $\ecvs(\G^{+F\cup\{xy\}},\R)=\ecvs(\G^{+F},\R).$
\label{fig:parallel-fictive2}}
\end{figure}

\begin{proof}
Let {$\sigma^*\in\mathcal{L}^c(\G^{+F\cup\{xy\}},\R)$} be a connected layout of minimum cost. {Consider the neighbor $v$ of $y$ such that $\sigma^*(v)$ is minimum. By the connectivity of $\sigma^*,$  $\sigma^*(v)<\sigma^*(y)$. Suppose without loss of generality that $v\in V_1$.} The case $v\in V_2$ is symmetric. 
%We construct a layout $\sigma$ of {$(\G^{+F\cup\{xy\}},\R)$} as in the proof of Lemma~\ref{lem_parallel_fictive}, that is: $\sigma=\langle r_1,\dots, r_k\rangle\odot \sigma^*[V_1]\odot \sigma^*[V_2\setminus\{x,y\}]$.
%We first observe that Claim~\ref{cl:parallel-fictive-1} and Claim~\ref{cl:parallel-fictive-2} apply to $\sigma$. 

\xchange{
As in the proof of Lemma~\ref{lem_parallel_fictive}, we define  the layout $\sigma_1=\sigma[V_1\cup \R]$ of  $(\G[V_1\cup\R]^{+F\cup\{xy\}},\R)$ and the layout  $\sigma_2=\sigma[V_2]$ of $(\G_2,\{x,y\})$,  where $\sigma=\langle r_1,\dots, r_k\rangle\odot \sigma^*[V_1]\odot \sigma^*[V_2\setminus\{x,y\}]$ is a layout of $(\G^{+F\cup\{xy\}},\R)$. 
We first observe that Claim~\ref{cl:parallel-fictive-1} 
applies to $\sigma_1$, $\sigma_2$ and $\sigma$ which are thereby connected layouts. We also remark that  the proof of Claim~\ref{cl:parallel-fictive-2} applies to $\sigma$ even in the presence of the fictive edge $xy$.
%\xtof{The fact that claim 2 applies is a bit abusive because of fictive edge $xy$.} 
So we have $\sigma\in\mathcal{L}^c(\G^{+F},\R)$ and $\ecost((\G^{+F},\R),\sigma)= \ecost((\G^{+F},\R),\sigma^*)$.
Following the proof of Lemma~\ref{lem_parallel_fictive}, we show that for every $v_1\in V_1\cup \R$, $S^{(e)}_{\sigma}(v_1)=S^{(e)}_{\sigma_1}(v_1)$ and that for every $v_2\in V_2\setminus\{x,y\}$, $S^{(e)}_{\sigma}(v_2)=S^{(e)}_{\sigma_2}(v_2)$.
So we proved that {$\ecvs(\G^{+F\cup\{xy\}},\R)=\max \left\{\ecvs(\G[V_1\cup\R]^{+F\cup\{xy\}},\R), \ecvs(\G_2^{+\emptyset},\{x,y\}) \right\}$}. Lemma~\ref{lem_parallel_fictive} allows us to conclude that {$\ecvs(\G^{+F\cup\{x,y\}},\R)=\ecvs(\G^{+F},\R)$} as claimed. In other words the fictive edge $xy$ is irrelevant  (see Figure~\ref{fig:parallel-fictive2}).
}
%Moreover the proof of Lemma \ref{lem_parallel_fictive} shows that {$\ecvs(\G^{+F\cup\{xy\}},\R)=\max \left\{\ecvs(\G[V_1\cup\R]^{+F\cup\{xy\}},\R), \ecvs(\G_2^{+\emptyset},\{x,y\}) \right\}$}. It follows that the fictive edge $xy$ is irrelevant (see Figure~\ref{fig:parallel-fictive2}) and that {$\ecvs(\G^{+F\cup\{x,y\}},\R)=\ecvs(\G^{+F},\R)$} as claimed.
\end{proof}

%--------------------------

\subsubsection{Series composition}

\begin{lemma} \label{lem_series_no_fictive}
Let $(\G^{+\emptyset},\R)$ be a rooted extended graph with $\R=\{x,y\}$ and such that $\G=\G_1\otimes \G_2$ with $\G_1=(G_1,(x,z))$ and $\G_2=(G_2,(z,y))$
(see Figure~\ref{fig:series}). Then

$$\ecvs(\G^{+\emptyset},\R)=\min \left\{\begin{array}{c} 
\max\left\{\ecvs(\tilde{\G}_1^{+\{zy\}},\R),\ecvs(\G_2^{+\emptyset},\R_2)\right\}\\ 
\\
\max\left\{\ecvs(\tilde{\G}_2^{+\{zx\}},\R),\ecvs(\G_1^{+\emptyset},\R_1)  \right\} \\
\end{array}\right\},
$$
where $\tilde{\G}_1$ (resp. $\tilde{\G}_2$) is obtained from $\G_1$ (resp. $\G_2$) by adding $y$ (resp. $x$) as an isolated vertex,  and where $\R_1=\{z,x\},$  $\R_2=\{z,y\}$.
\end{lemma}
\begin{proof}
Let $\sigma^*\in\mathcal{L}^c(\G,\R)$ be a connected layout of minimum cost. {Consider the neighbor $v$ of $z$ so that $\sigma^*(v)$ is minimized. By the connectivity of $\sigma^*,$  $\sigma^*(v)<\sigma^*(z)$. Suppose without loss of generality that $v\in V_1$.} The case $v\in V_2$ is symmetric. 
%From $\sigma^*,$  we define a layout $\sigma$ of $(\G^{+\emptyset},\R)$ as follows: $\sigma=\langle x,y\rangle \odot\sigma^*[V_1\setminus\{x\}]\odot \sigma^*[V_2\setminus\{y,z\}]$ (see Figure~\ref{fig:series-rearrangement}).
\xchange{
From $\sigma^*,$  we define $\sigma_1=\sigma[V_1\cup\R]$ a layout of $(\ecvs(\tilde{\G}_1^{+\{zy\}},\R)$, $\sigma_2=\langle z,y\rangle \odot \sigma[V_2\setminus\{y,z\}]$ a layout of $(\G_2^{+\emptyset},\R_2)$ and $\sigma=\langle x,y\rangle \odot\sigma^*[V_1\setminus\{x\}]\odot \sigma^*[V_2\setminus\{y,z\}]$ a layout of $(\G^{+\emptyset},\R)$ (see Figure~\ref{fig:series-rearrangement}).
}

\begin{figure}[h]
\centering
\begin{tikzpicture}[scale=0.9]
   \tikzstyle{vertex}=[fill,circle,minimum size=0.2cm,inner sep=0pt]
   \tikzstyle{vertexRoot}=[fill,red,rectangle,minimum size=0.2cm,inner sep=0pt]
   \tikzstyle{vertexRoot2}=[fill,draw,circle,minimum size=0.2cm,inner sep=0pt]
   \tikzstyle{vertex2}=[draw,circle,minimum size=0.2cm,inner sep=0pt]
   \tikzstyle{diamondb}=[fill=red!20,draw,diamond,minimum size=0.25cm,inner sep=0pt]
   \tikzstyle{diamondw}=[fill=blue!20,draw,diamond,minimum size=0.25cm,inner sep=0pt]

%\sigma*
\draw[thin,->,arrows={-latex'}] (-5,0) -- (5.5,0);
\node[anchor=east] at (-5,0) {$\sigma^*$};

\foreach \i/\posnd in {1/-4.5,2/-4}
{\node[vertexRoot] (nd\i) at (\posnd,0){};}

\foreach \i/\posnd in {12/1}
{\node[vertexRoot2] (nd\i) at (\posnd,0){};}

\node[anchor=north] at (-4.5,0.6) {$x$};
\node[anchor=north] at (-4,0.6) {$y$};
\node[anchor=north] at (1,0.6) {$z$};
\node[anchor=south] at (-2.5,0.15) {$v$};
\node[anchor=south] at (-1,0.07) {$v_2$};

\foreach \i/\posnd in {3/-3.5,4/-3,5/-2.5,6/-2,9/-0.5,10/0,15/2.5,16/3,18/4}
{\node[diamondb] (nd\i) at (\posnd,0){};}
\foreach \i/\posnd in {7/-1.5,8/-1,11/0.5,13/1.5,14/2,17/3.5,19/4.5,20/5}
{\node[diamondw] (nd\i) at (\posnd,0){};}

%\sigma
\draw[thin,->,arrows={-latex'}] (-5,-1.5) -- (5.5,-1.5);
\node[anchor=east] at (-5,-1.5) {$\sigma$};

\foreach \i/\posnd in {1/-4.5,2/-4}
{\node[vertexRoot] (nds\i) at (\posnd,-1.5){};}
\foreach \i/\posnd in {12/-0.5}
{\node[vertexRoot2] (nds\i) at (\posnd,-1.5){};}

\node[anchor=north] at (-4.5,-1.7) {$x$};
\node[anchor=north] at (-4,-1.7) {$y$};
\node[anchor=north] at (-0.5,-1.7) {$z$};
\node[anchor=north] at (-2.5,-1.7) {$v$};
\node[anchor=north] at (2,-1.7) {$v_2$};

\foreach \i/\posnd in {3/-3.5,4/-3,5/-2.5,6/-2,9/-1.5,10/-1,15/0,16/0.5,18/1}
{\node[diamondb] (nds\i) at (\posnd,-1.5){};}
\foreach \i/\posnd in {7/1.5,8/2,11/2.5,13/3,14/3.5,17/4,19/4.5,20/5}
{\node[diamondw] (nds\i) at (\posnd,-1.5){};}

\foreach \i in {1,2,3,4,5,6,7,8,9,10,11,12,13,14,15,16,17,18,19,20}
{\draw[gray!60,thick,dotted] (nd\i) -- (nds\i);
}

\draw[ thick,-,>=latex] (nd12) to[bend left=40] (nd5);
\draw[ thick,-,>=latex] (nd12) to[bend left=30] (nd8);
\draw[ thick,-,>=latex] (nds5) to[bend left=40] (nds12);
\draw[ thick,-,>=latex] (nds12) to[bend left=30] (nds8);

\end{tikzpicture}
\caption{Rearranging a layout $\sigma^*$ of $\G=\G_1\otimes\G_2$ of minimum cost into $\sigma=\langle x,y\rangle \odot\sigma^*[V_1\setminus\{x\}]\odot \sigma^*[V_2\setminus\{y,z\}]$. Red diamond vertices belong to $V_1\setminus\{x\}$, blue diamond vertices belong to $V_2\setminus\{y,z\}$ and red square vertices are the roots. Observe that the path $v,z,v_2$ certifies that $v\in S^{(e)}_{\sigma^*}(v_2)$. But as $\sigma(z)<\sigma(v_2)$, {$v\notin S^{(e)}_{\sigma}(v_2)$}. Instead we have that $z\in S^{(e)}_{\sigma}(v_2)$.
\label{fig:series-rearrangement}}
\end{figure}
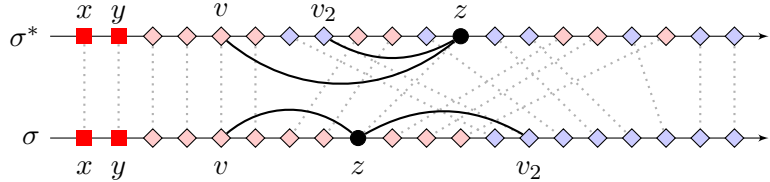

\begin{claim} \label{cl:series1}
$\sigma_1\in\mathcal{L}^c(\ecvs(\tilde{\G}_1^{+\{zy\}},\R)$, $\sigma_2\in\mathcal{L}^c(\G_2^{+\emptyset},\R_2)$ and $\sigma\in\mathcal{L}^c(\G^{+\emptyset},\R)$. 
\end{claim}
\begin{proofclaim}
Let $v_1$ be a vertex of $V_1$ distinct from $x$ and $z$. Every neighbor of $v_1$ belongs to $V_1$. As the relative ordering between vertices of $V_1$ is left unchanged, vertex $v_1$ has a neighbor prior to it in $\sigma$. 
\xchange{
It follows that $\sigma_1\in\mathcal{L}^c(\ecvs(\tilde{\G}_1^{+\{zy\}},\R)$.
}
Suppose that $v_2$ is a vertex of $V_2$ distinct from $y$ and $z$. Then as in the previous case,  every neighbor of $v_2$ belongs to $V_2$. 
As the relative ordering in $\sigma_2$ and $\sigma$ between vertices of $V_2$ has only been modified by moving $z$ ahead, $v_2$ has a neighbor prior to it in $\sigma$. 
\xchange{
It follows that $\sigma_2\in\mathcal{L}^c(\G_2^{+\emptyset},\R_2)$.
}
To prove that $\sigma$ is connected, we are left with vertex $z$. By assumption, $z$ has a neighbor $v\in V_1$ such that $\sigma^*(v)<\sigma^*(z),$  implying that $\sigma(v)<\sigma(z)$. As every vertex has a neighbor prior to it in $\sigma,$  the layout $\sigma$ is connected.
\end{proofclaim}

\begin{figure}[h]
\centering
\begin{tikzpicture}[scale=1.0]
   \tikzstyle{vertex}=[fill,circle,minimum size=0.2cm,inner sep=0pt]
   \tikzstyle{vertexRoot}=[fill,red,rectangle,minimum size=0.15cm,inner sep=0pt]

%---- G------
\filldraw[fill=red!20]
  (2,-2) .. controls (1.3,-1) .. (2,0) .. controls (2.7,-1) .. (2,-2)
  -- cycle ;
\node[] at (2,-1) {$G_1$};

\filldraw[fill=blue!20]
  (2,0) .. controls (1.3,1) .. (2,2) .. controls (2.7,1) .. (2,0)
  -- cycle ;
\node[] at (2,1) {$G_2$};

%\filldraw[fill=blue!20]
%  (0,0) .. controls (1.5,1) .. (0,2) .. controls (0.5,1) .. (0,0)
%  -- cycle ;
%\node[] at (0.7,1) {$G_2$};

\node[vertex]  (z) at (2,0){};
\node[vertexRoot] (x) at (2,-2){};
\node[vertexRoot] (y) at (2,2){};

%\foreach \k in {0,1,2}{
%	\node[vertexRoot] (r\k) at (\k-1,3){};
%	\draw[red,thick,dashed] (y) -- (r\k);
%}

%\node[above] at (r0) {$r_1$};
%\node[above] at (r1) {$r_i$};
%\node[above] at (r2) {$r_k$};
\node[left] at (z) {$z$};
\node[below] at (x) {$x$};
\node[above] at (y) {$y$};

\draw[black,very thick,->] (3.25,0) -- (3.75,0) ;

%---- G1------
\filldraw[fill=red!20]
  (5,-1.5) .. controls (4.3,-0.5) .. (5,0.5) .. controls (5.7,-0.5) .. (5,-1.5)
  -- cycle ;
\node[] at (5,-0.5) {$G_1$};

\node[vertexRoot]  (yy) at (5,1.5){};
\node[vertexRoot] (xx) at (5,-1.5){};
\node[vertex] (zz) at (5,0.5){};

\node[left] at (zz) {$z$};
\node[below] at (xx) {$x$};
\node[above] at (yy) {$y$};

\draw[red,dashed] (zz) -- (yy) ;

\draw[black,very thick] (6.3,0) -- (6.7,0) ;
\draw[black,very thick] (6.5,-0.2) -- (6.5,0.2) ;

%---- G2------

\filldraw[fill=blue!20]
  (8,-1) .. controls (7.3,0) .. (8,1) .. controls (8.7,0) .. (8,-1)
  -- cycle ;
\node[] at (8,0) {$G_2$};

\node[vertexRoot]  (yyy) at (8,1){};
\node[vertexRoot] (zzz) at (8,-1){};

\node[left] at (yyy) {$y$};
\node[below] at (zzz) {$z$};

\end{tikzpicture}
\caption{Decomposition of an extended graph resulting from a series composition.
\label{fig:series}}
\end{figure}
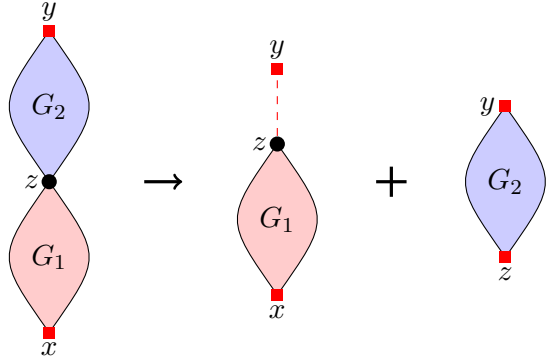

\begin{claim} \label{cl:series2}
{$\ecost((\G^{+\emptyset},\R),\sigma)= \ecost((\G^{+\emptyset},\R),\sigma^*)$.} 
\end{claim}

%We first consider a vertex $v_1\in V_1$. Observe that as $\sigma^*[V_1]=\sigma[V_1],$  we have $S^{(e)}_{\sigma^*}(v_1)\cap V_1=S^{(e)}_{\sigma}(v_1)\cap V_1$. Suppose that $\sigma(v_1)>\sigma(z)$. As $z$ separates $v_1$ from the vertices of $V_2,$  we have that $S^{(e)}_{\sigma}(v_1)=S^{(e)}_{\sigma^*}(v_1)$. Suppose now that $\sigma(v_1)\leq \sigma(z)$ and that there exists a vertex $v\in V_2\setminus\{z\}$ such that $v\in S^{(e)}_{\sigma^*}(v_1)$. As in $\sigma,$  every vertex of $V_2\setminus\{z,y\}$ occurs after the vertices of $V_1,$  we have that $y\in S^{(e)}_{\sigma}(v_1)$ but not $v\notin S^{(e)}_{\sigma}(v_1)$. It follows that $|S^{(e)}_{\sigma}(v_1)|\leq |S^{(e)}_{\sigma^*}(v_1)|$.

\begin{proofclaim}
We first consider a vertex $v_1\in V_1$. {By construction, we have $\sigma^*[V_1]=\sigma[V_1]$ and for every vertex $v_2\in V_2\setminus\{x,y\},$  $\sigma(v_1)\leqslant \sigma(v_2)$. It follows that $S^{(e)}_{\sigma}(v_1)\subseteq V_1\cup\R$.} Suppose that $\sigma(v_1)>\sigma(z)$. As $z$ separates $v_1$ from the vertices of $V_2,$  we have that $S^{(e)}_{\sigma}(v_1)=S^{(e)}_{\sigma^*}(v_1)$. Suppose now that $\sigma(v_1)\leq \sigma(z)$ 
{
and let $v\in V_1$ a vertex that belongs to $S^{(e)}_{\sigma}(v_1)$. Then, as $z$ is a cut vertex of $\G^{+\emptyset},$  there exists a $(v,v_1)$-path $P$ in $G_1$ such that every vertex $v'\in P$ distinct from $v$ satisfies $\sigma(v_1)\leq \sigma(v')$. As $\sigma^*[V_1]=\sigma[V_1],$  the path $P$ certifies that $v\in S^{(e)}_{\sigma^*}(v_1),$  implying that $|S^{(e)}_{\sigma}(v_1)|\leq |S^{(e)}_{\sigma^*}(v_1)|$.
}

Let us now consider a vertex $v_2\in V_2\setminus\{y,z\}$.
Observe that, as $z$ is a cut vertex, $\sigma^*(z)<\sigma^*(v_2)$ implies that $S^{(e)}_{\sigma^*}(v_2)\subseteq V_2$. As $\sigma[V_2\setminus\{z\}]=\sigma^*[V_2\setminus\{z\}],$  every vertex $u\in V_2\setminus\{z\}$ that belongs to $S^{(e)}_{\sigma}(v_2)$ also belongs to $S^{(e)}_{\sigma^*}(v_2)$. Suppose that $z\in S^{(e)}_{\sigma}(v_2)$. Then there exists a $(z,v_2)$-path $P$ such that every internal vertex $u$ of $P$ satisfies $\sigma(v_2)<\sigma(u)$. 
As $\sigma[V_2\setminus\{z\}]=\sigma^*[V_2\setminus\{z\}],$  we also have $\sigma^*(v_2)<\sigma^*(v)$. Let us distinguish two cases:
\begin{itemize}
\item If $\sigma^*(z)<\sigma^*(v_2),$  then  $z\in S^{(e)}_{\sigma}(v_2)$ implying that $S^{(e)}_{\sigma}(v_2)\subseteq S^{(e)}_{\sigma^*}(v_2)$. 
\item Otherwise, $\sigma^*(z)>\sigma^*(v_2)$ and then $z\notin S^{(e)}_{\sigma^*}(v_2)$. But in that case, let us recall that by assumption the first neighbor $v$ of $z$ in $\sigma^*$ belongs to $V_1$. It follows that $v\in S^{(e)}_{\sigma^*}(v_2)$. As we argue that $S^{(e)}_{\sigma}(v_2)\subseteq V_2,$  $v\in S^{(e)}_{\sigma^*}(v_2)\setminus S^{(e)}_{\sigma}(v_2)$. So $v$ is a replacement vertex for $z$ in $S^{(e)}_{\sigma^*}(v_2),$  implying that $|S^{(e)}_{\sigma}(v_2)|\leq |S^{(e)}_{\sigma^*}(v_2)|$.
\end{itemize}

So we proved that $\ecost((\G^{+\emptyset},\R),\sigma)\leq \ecost((\G^{+\emptyset},\R),\sigma^*)$. As by Claim~\ref{cl:series1}, $\sigma\in\mathcal{L}^c(\G^{+F},\R),$  the optimality of $\sigma^*$ implies that $\ecost((\G^{+\emptyset},\R),\sigma)= \ecost((\G^{+\emptyset},\R),\sigma^*)$.
\end{proofclaim}

Let us now conclude the proof. 
Claim~\ref{cl:series1}, Claim~\ref{cl:series2} and the optimality of $\sigma^*$ imply that $\sigma$ is a connected layout of $\G^{+F}$ of minimum cost.
By Claim~\ref{cl:series1}, we have $\sigma_1=\sigma[V_1\cup \R]\in\mathcal{L}^c(\tilde{\G}_1^{+\{xy\}},\R)$. 
Observe that in the extended graph $\tilde{\G}_1^{+\{xy\}},$  the fictive edge $xy$ simulates every $(z,y)$-path of $\G$ whose internal vertices belong to $V_2$.
It follows that for every vertex $v_1\in V_1\cup\R,$  $S^{(e)}_{\sigma}(v_1)=S^{(e)}_{\sigma_1}(v_1)$. Likewise, by Claim~\ref{cl:series1}, we know that $\sigma_2=\sigma[V_2]\in\mathcal{L}^c(\G_2^{+\emptyset},\{z,y\})$.  As $z$ separates the vertices of $V_2$ from the other vertices, for every $v_2\in V_2\setminus\{z\},$  we have $S^{(e)}_{\sigma}(v_2)=S^{(e)}_{\sigma_2}(v_2)$. 
It follows that $\ecvs(\G^{+\emptyset},\R)=\max\left\{\ecvs(\tilde{\G}_1^{+\{zy\}},\R),\ecvs(\G_2^{+\emptyset},\R_2)\right\}$.
\end{proof}

Let us now consider a rooted extended graph $(\G^{+F},\R)$ with $\R=\{x,r_1\dots, r_k\}$ and $(G[V\setminus\{r_1,\dots, r_k\}],(x,y))=\G_1\otimes \G_2$  (see Figure~\ref{fig:series-rooted}). The fact that $y$ is not a root vertex forces that in  a connected layout starting from the root $x$ forces the vertex $z$ to be the first vertex of $V_2$. Using this observation, one can apply the same arguments than in the proof of Lemma~\ref{lem_series_no_fictive}.

\begin{lemma} \label{lem_series_fictive}
Let $(\G^{+F},\R)$ be a rooted extended graph such that $\R=\{x,r_1\dots, r_k\}$ (with $k>0$), $r_1,\dots r_k$ are isolated vertices (in $\G$), $F=\{yr_i\mid i\in[k]\}$ and $(G[V\setminus\{r_1,\dots, r_k\}],(x,y))=\G_1\otimes \G_2$ with $\G_1=(G_1,(x,z))$ and $\G_2=(G_2,(z,y))$. Then

$$\ecvs(\G^{+F},\R)
= \max \left\{\ecvs(\G[V_1\cup\R]^{+F'},\R), \ecvs(\G[V_2\cup\R']^{+F},\R') \right\},
$$
where $F'=\{zr_i\mid i\in[k]\}$ and $\R'=\{z,r_1,\dots, r_k\}$.
\end{lemma}

\begin{proof}
Let $\sigma^*\in\mathcal{L}^c(\G^{+F},\R)$ be a connected layout of minimum cost. 
%From $\sigma^*,$  we define a layout $\sigma$ of $(\G^{+F},\R)$ as follows: $\sigma=\langle r_1,\dots, r_k\rangle \odot \sigma^*[V_1]\odot \sigma^*[V_2\setminus\{z\}]$ (see Figure~\ref{fig:series-rearrangement-roots}). 
\xchange{
From $\sigma^*,$  we define: $\sigma_1=\sigma[V_1\cup\R]$ a layout of $(\G[V_1\cup\R]^{+F'},\R)$, $\sigma_2=\langle z,y\rangle \odot \sigma[V_2\setminus\{y,z\}]$ a layout of $(\G[V_2\cup\R']^{+F},\R')$ and $\sigma=\langle r_1,\dots, r_k\rangle \odot\sigma^*[V_1]\odot \sigma^*[V_2\setminus\{z\}]$ (see Figure~\ref{fig:series-rearrangement-roots}) a layout of $(\G^{+F},\R)$.
}
We observe that in this case, as $y$ is not a root, $\sigma^*(z)\leq\sigma^*(v_2)$ for every vertex $v_2\in V_2$.

\begin{figure}[h]
\centering
\begin{tikzpicture}[scale=1.0]
   \tikzstyle{vertex}=[fill,circle,minimum size=0.2cm,inner sep=0pt]
   \tikzstyle{vertexRoot}=[fill,red,rectangle,minimum size=0.2cm,inner sep=0pt]
   \tikzstyle{vertexRoot2}=[fill,draw,circle,minimum size=0.2cm,inner sep=0pt]
   \tikzstyle{vertex2}=[draw,circle,minimum size=0.2cm,inner sep=0pt]
   \tikzstyle{diamondb}=[fill=red!20,draw,diamond,minimum size=0.25cm,inner sep=0pt]
   \tikzstyle{diamondw}=[fill=blue!20,draw,diamond,minimum size=0.25cm,inner sep=0pt]

%----------------------
%\sigma*
\draw[thin,->,arrows={-latex'}] (-6,0) -- (5.5,0);
\node[anchor=east] at (-6,0) {$\sigma^*$};

\foreach \i/\posnd in {1/-5.5,2/-5,3/-4.5}{
	\node[vertexRoot] (r\i) at (\posnd,0){};
	\node[anchor=north] at (\posnd,0.6) {$r_{\i}$};
	}	
\foreach \i/\posnd in {2/-4}
{\node[vertexRoot] (nd\i) at (\posnd,0){};}

\foreach \i/\posnd in {7/-1.5}
{\node[vertexRoot2] (nd\i) at (\posnd,0){};}

\node[anchor=north] at (-4,0.6) {$x$};
\node[anchor=north] at (-1.5,0.6) {$z$};
\node[anchor=north] at (3.5,0.6) {$y$};
%\node[anchor=south] at (-2.5,0.15) {$v$};
%\node[anchor=south] at (-1,0.07) {$v_2$};

%G1
\foreach \i/\posnd in {3/-3.5,4/-3,5/-2.5,6/-2,9/-0.5,10/0, 12/1,15/2.5,16/3,18/4}
	{\node[diamondb] (nd\i) at (\posnd,0){};}

%1/-4.5
%G2
\foreach \i/\posnd in {8/-1,11/0.5,13/1.5,14/2,17/3.5,19/4.5,20/5}
	{\node[diamondw] (nd\i) at (\posnd,0){};}

%----------------------
%\sigma
\draw[thin,->,arrows={-latex'}] (-6,-1.5) -- (5.5,-1.5);
\node[anchor=east] at (-6,-1.5) {$\sigma$};

\foreach \i/\posnd in {1/-5.5,2/-5,3/-4.5}{
	\node[vertexRoot] (rs\i) at (\posnd,-1.5){};
	\node[anchor=north] at (\posnd,-1.7) {$r_{\i}$};
	}

\foreach \i/\posnd in {1/-4.5,2/-4}
{\node[vertexRoot] (nds\i) at (\posnd,-1.5){};}
\foreach \i/\posnd in {7/-1.5}
{\node[vertexRoot2] (nds\i) at (\posnd,-1.5){};}

\node[anchor=north] at (-4,-1.7) {$x$};
\node[anchor=north] at (4,-1.7) {$y$};
\node[anchor=north] at (-1.5,-1.7) {$z$};
%\node[anchor=north] at (-2.5,-1.7) {$v$};
%\node[anchor=north] at (2,-1.7) {$v_2$};

\foreach \i/\posnd in {3/-3.5,4/-3,5/-2.5,6/-2, 9/-1, 10/-0.5,12/0,15/0.5,16/1, 18/1.5}
{\node[diamondb] (nds\i) at (\posnd,-1.5){};}

\foreach \i/\posnd in {8/2,11/2.5,13/3,14/3.5,17/4,19/4.5,20/5}
{\node[diamondw] (nds\i) at (\posnd,-1.5){};}

\foreach \i in {2,3,4,5,6,7,8,9,10,11,12,13,14,15,16,17,18,19,20}
{\draw[gray!60,thick,dotted] (nd\i) -- (nds\i);
}
\foreach \i in {1,2,3}
{\draw[gray!60,thick,dotted] (r\i) -- (rs\i);
}

%\draw[ thick,-,>=latex] (nd12) to[bend left=40] (nd5);
%\draw[ thick,-,>=latex] (nd12) to[bend left=30] (nd8);
%\draw[ thick,-,>=latex] (nds5) to[bend left=40] (nds12);
%\draw[ thick,-,>=latex] (nds12) to[bend left=30] (nds8);

\end{tikzpicture}
\caption{Rearranging a layout $\sigma^*$ of minimum cost into $\sigma=\langle r_1,r_2,r_3,x\rangle \odot\sigma^*[V_1\setminus\{x\}]\odot \sigma^*[V_2\setminus\{z\}]$. Red diamond vertices belong to $V_1\setminus\{x\}$, blue diamond vertices belong to $V_2\setminus\{z\}$, red square vertices are the roots.
\label{fig:series-rearrangement-roots}}
\end{figure}
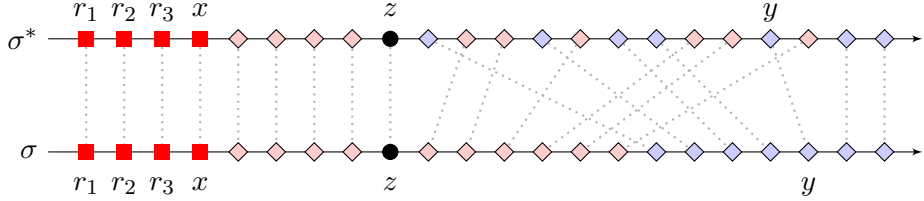

\begin{claim} \label{cl:series-roots1}
$\sigma_1\in\mathcal{L}^c(\G[V_1\cup\R]^{+F'},\R)$,
$\sigma_2\in\mathcal{L}^c(\G[V_2\cup\R']^{+F},\R')$ and 
$\sigma\in\mathcal{L}^c(\G^{+F},\R)$. 
\end{claim}

\begin{proofclaim}
Let $v_1$ be a vertex of $V_1$ distinct from $x$. Every neighbor of $v_1$ belongs to $V_1$. As the relative ordering between vertices of $V_1$ is left unchanged, vertex $v_1$ has a neighbor prior to it in $\sigma$. 
\xchange{
It follows that $\sigma_1\in\mathcal{L}^c(\G[V_1\cup\R]^{+F'},\R)$.
}
Suppose that $v_2$ is a vertex of $V_2$ distinct from $z$. Then as in the previous case,  every neighbor of $v_2$ belongs to $V_2$. As the relative ordering between vertices of $V_2$ is left unchanged in $\sigma_2$ and $\sigma$, vertex $v_2$ has a neighbor prior to it in $\sigma$. It follows that 
\xchange{
$\sigma_2\in\mathcal{L}^c(\G[V_2\cup\R']^{+F},\R')$ and $\sigma\in\mathcal{L}^c(\G^{+F},\R)$.
}
%the layout $\sigma$ is connected.
\end{proofclaim}

\begin{figure}[h]
\centering
\begin{tikzpicture}[scale=0.9]
   \tikzstyle{vertex}=[fill,circle,minimum size=0.2cm,inner sep=0pt]
   \tikzstyle{vertexRoot}=[fill,red,rectangle,minimum size=0.15cm,inner sep=0pt]

%---- G------
\filldraw[fill=red!20]
  (1.5,-2) .. controls (0.8,-1) .. (1.5,0) .. controls (2.2,-1) .. (1.5,-2)
  -- cycle ;
\node[] at (1.5,-1) {$G_1$};

\filldraw[fill=blue!20]
  (1.5,0) .. controls (0.8,1) .. (1.5,2) .. controls (2.2,1) .. (1.5,0)
  -- cycle ;
\node[] at (1.5,1) {$G_2$};

%\filldraw[fill=blue!20]
%  (0,0) .. controls (1.5,1) .. (0,2) .. controls (0.5,1) .. (0,0)
%  -- cycle ;
%\node[] at (0.7,1) {$G_2$};

\node[vertex]  (z) at (1.5,0){};
\node[vertexRoot] (x) at (1.5,-2){};
\node[vertex] (y) at (1.5,2){};

\foreach \k in {0,1,2}{
	\node[vertexRoot] (r\k) at (\k+0.5,3){};
	\draw[red,thick,dashed] (y) -- (r\k);
}

\draw[black,thick,dotted] (0.7,3) -- (1.3,3) ;
\draw[black,thick,dotted] (1.7,3) -- (2.3,3) ;

\node[above] at (r0) {$r_1$};
\node[above] at (r1) {$r_i$};
\node[above] at (r2) {$r_k$};
\node[left] at (z) {$z$};
\node[below] at (x) {$x$};
\node[left] at (y) {$y$};

\draw[black,very thick,->] (3.25,0) -- (3.75,0) ;

%---- G1------
\filldraw[fill=red!20]
  (5.5,-1.5) .. controls (4.8,-0.5) .. (5.5,0.5) .. controls (6.2,-0.5) .. (5.5,-1.5)
  -- cycle ;
\node[] at (5.5,-0.5) {$G_1$};

\node[vertexRoot] (xx) at (5.5,-1.5){};
\node[vertex] (zz) at (5.5,0.5){};

\foreach \k in {0,1,2}{
	\node[vertexRoot] (rr\k) at (\k+4.5,1.5){};
	\draw[red,thick,dashed] (zz) -- (rr\k);
}

\draw[black,thick,dotted] (4.7,1.5) -- (5.3,1.5) ;
\draw[black,thick,dotted] (5.7,1.5) -- (6.3,1.5) ;

\node[above] at (rr0) {$r_1$};
\node[above] at (rr1) {$r_i$};
\node[above] at (rr2) {$r_k$};
\node[left] at (zz) {$z$};
\node[below] at (xx) {$x$};

\draw[black,very thick] (7.3,0) -- (7.7,0) ;
\draw[black,very thick] (7.5,-0.2) -- (7.5,0.2) ;

%---- G2------

\filldraw[fill=blue!20]
  (9.5,-1.5) .. controls (8.8,-0.5) .. (9.5,0.5) .. controls (10.2,-0.5) .. (9.5,-1.5)
  -- cycle ;
\node[] at (9.5,-0.5) {$G_2$};

\node[vertex]  (yyy) at (9.5,0.5){};
\node[vertexRoot] (zzz) at (9.5,-1.5){};

\foreach \k in {0,1,2}{
	\node[vertexRoot] (rrr\k) at (\k+8.5,1.5){};
	\draw[red,thick,dashed] (yyy) -- (rrr\k);
}

\draw[black,thick,dotted] (8.7,1.5) -- (9.3,1.5) ;
\draw[black,thick,dotted] (9.7,1.5) -- (10.3,1.5) ;

\node[above] at (rrr0) {$r_1$};
\node[above] at (rrr1) {$r_i$};
\node[above] at (rrr2) {$r_k$};
\node[left] at (yyy) {$y$};
\node[below] at (zzz) {$z$};

\end{tikzpicture}
\caption{Decomposition of an extended graph resulting from a series composition. A connected layout starting at $x$ places $z$ as the first vertex of $V_2$.
\label{fig:series-rooted}}
\end{figure}
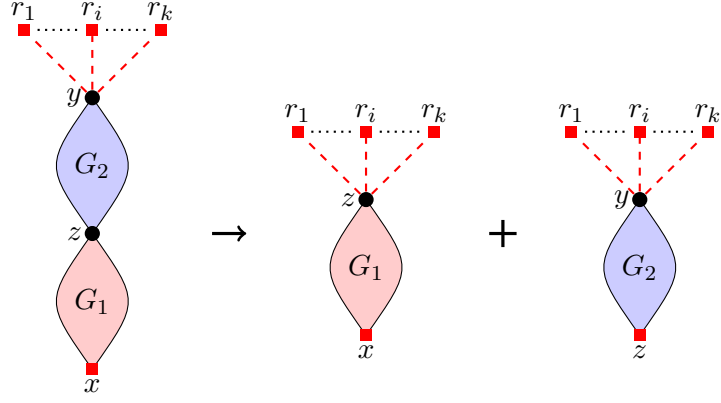

\begin{claim} \label{cl:series-roots2}
$\ecost((\G^{+F},\R),\sigma)= \ecost((\G^{+F},\R),\sigma^*)$. 
\end{claim}

\begin{proofclaim}
We first consider a vertex $v_1\in V_1$. Observe that $\sigma^*(z)\leq\sigma^*(v_2)$ for every vertex $v_2\in V_2$. It follows that $S^{(e)}_{\sigma^*}(v_1)\cap(V_2\setminus\{z\})=\emptyset$. 
The fact that $\sigma^*[V_1\cup\R]=\sigma[V_1\cup\R]$ implies $S^{(e)}_{\sigma}(v_1)=S^{(e)}_{\sigma^*}(v_1)$. 
Similar arguments hold for every vertex $v_2\in V_2\setminus\{z\}$. As $z$ separates vertices of $V_2$ from vertices of $V_1$ and as $\sigma^*(x)<\sigma^*(z)<\sigma^*(v_2),$  we have $S^{(e)}_{\sigma^*}(v_2)\cap(V_1\setminus\{x,z\})=\emptyset$. 
The fact that $\sigma^*[V_2\cup\R]=\sigma[V_2\cup\R]$ implies $S^{(e)}_{\sigma}(v_2)=S^{(e)}_{\sigma^*}(v_2)$. Thereby $\ecost((\G^{+F},\R),\sigma)= \ecost((\G^{+F},\R),\sigma^*)$. 
\end{proofclaim}

%\medskip
Let us now conclude the proof of the lemma. Claim~\ref{cl:series-roots1}, Claim~\ref{cl:series-roots2} and the optimality of $\sigma^*$ imply that $\sigma$ is a connected layout of $G^{+F}$ of minimum cost. 
%From the proof of Claim~\ref{cl:series-roots1}, we deduce that $\sigma_1=\sigma[V_1\cup \R]\in\mathcal{L}^c(\G[V_1\cup\R]^{+F'},\R)$. 
From Claim~\ref{cl:series-roots1}, we have that $\sigma_1=\sigma[V_1\cup \R]\in\mathcal{L}^c(\G[V_1\cup\R]^{+F'},\R)$. 
{Observe that in the extended graph $\G[V_1\cup\R]^{+F'},$  for $i\in[k],$  the fictive edge $zr_i\in F'$ simulates every simple extended $(z,r_i)$-path in $\G^{+F}$.}
It follows that for every vertex $v_1\in V_1\cup\R,$  $S^{(e)}_{\sigma}(v_1)=S^{(e)}_{\sigma_1}(v_1)$. 
%Likewise, from the proof of Claim~\ref{cl:series-roots1}, we deduce that $\sigma_2=\sigma[V_2\cup\R']\in\mathcal{L}^c(\G[V_2\cup\R']^{+F},\R')$.  
Likewise, from  Claim~\ref{cl:series-roots1}, we know that $\sigma_2=\sigma[V_2\cup\R']\in\mathcal{L}^c(\G[V_2\cup\R']^{+F},\R')$.  
As noticed before, $z$ separates the vertices of $V_2$ from vertices of $V_1$. It follows for every $v_2\in V_2\setminus\{z\},$  we have $S^{(e)}_{\sigma}(v_2)=S^{(e)}_{\sigma_2}(v_2),$  completing the proof.
\end{proof}

If there is a fictive edge between the two terminal vertices $x$ and $y$ of $\G$, then , as in Lemma~\ref{lem_series_fictive}, a connected layout starting at $x$ may visit $y$ before $z$. In this case, we obtain the following lemma.

\begin{lemma} \label{lem_series_fictive2}
Let $(\G^{+F\cup\{xy\}},\R)$ be a rooted extended graph such that $\R=\{x,r_1\dots, r_k\}$ (with $k>0$), $r_1,\dots r_k$ are isolated vertices (in $\G$), {$F=\{yr_i\mid i\in[k]\}$} and $(G[V\setminus\{r_1,\dots, r_k\}],(x,y))=\G_1\otimes \G_2$ with $\G_1=(G_1,(x,z))$ and $\G_2=(G_2,(z,y))$. Then

$$\ecvs(\G^{+F\cup\{xy\}},\R)
= \max \left\{\ecvs(\G[V_1\cup\R]^{+F'\cup\{xz\}},\R), \ecvs(\G[V_2\cup\R']^{+F\cup\{xy\}},\R') \right\},
$$
where $\R'=\{z,r_1,\dots, r_k,x\}$ and $F'=\{zr_i\mid i\in[k]\}$.
\end{lemma}
\begin{proof}
We proceed as in the proof of Lemma~\ref{lem_series_fictive}. We transform a layout $\sigma^*\in\mathcal{L}^c(\G^{+F\cup\{xy\}}),\R)$ of minimum cost into the layout $\sigma=\langle r_1,\dots, r_k\rangle \odot\sigma^*[V_1]\odot \sigma^*[V_2\setminus\{z\}]$ (see Figure~\ref{fig:series-rearrangement-roots}). Observe that since the solid graph $G$ is the same as in Lemma~\ref{lem_series_fictive}, Claim~\ref{cl:series-roots1} applies and thereby $\sigma\in\mathcal{L}^c(\G^{+F\cup\{xy\}},\R)$.

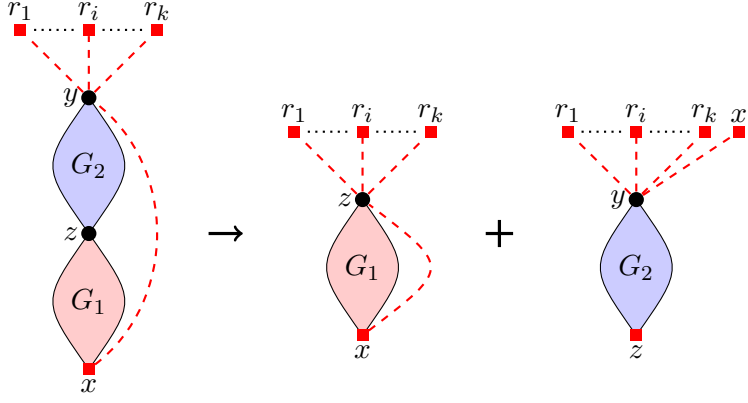
\begin{figure}[h]
\centering
\begin{tikzpicture}[scale=0.9]
   \tikzstyle{vertex}=[fill,circle,minimum size=0.2cm,inner sep=0pt]
   \tikzstyle{vertexRoot}=[fill,red,rectangle,minimum size=0.15cm,inner sep=0pt]

%---- G------
\filldraw[fill=red!20]
  (1.5,-2) .. controls (0.8,-1) .. (1.5,0) .. controls (2.2,-1) .. (1.5,-2)
  -- cycle ;
\node[] at (1.5,-1) {$G_1$};

\filldraw[fill=blue!20]
  (1.5,0) .. controls (0.8,1) .. (1.5,2) .. controls (2.2,1) .. (1.5,0)
  -- cycle ;
\node[] at (1.5,1) {$G_2$};

\node[vertex]  (z) at (1.5,0){};
\node[vertexRoot] (x) at (1.5,-2){};
\node[vertex] (y) at (1.5,2){};

\draw[red,thick,dashed] (x) .. controls (2.8,-1) and (2.8,1) .. (y) ;

\foreach \k in {0,1,2}{
	\node[vertexRoot] (r\k) at (\k+0.5,3){};
	\draw[red,thick,dashed] (y) -- (r\k);
}

\draw[black,thick,dotted] (0.7,3) -- (1.3,3) ;
\draw[black,thick,dotted] (1.7,3) -- (2.3,3) ;

\node[above] at (r0) {$r_1$};
\node[above] at (r1) {$r_i$};
\node[above] at (r2) {$r_k$};
\node[left] at (z) {$z$};
\node[below] at (x) {$x$};
\node[left] at (y) {$y$};

\draw[black,very thick,->] (3.25,0) -- (3.75,0) ;

%---- G1------
\filldraw[fill=red!20]
  (5.5,-1.5) .. controls (4.8,-0.5) .. (5.5,0.5) .. controls (6.2,-0.5) .. (5.5,-1.5)
  -- cycle ;
\node[] at (5.5,-0.5) {$G_1$};

\node[vertexRoot] (xx) at (5.5,-1.5){};
\node[vertex] (zz) at (5.5,0.5){};

\draw[red,thick,dashed] (xx) .. controls (6.8,-0.5) .. (zz) ;

\foreach \k in {0,1,2}{
	\node[vertexRoot] (rr\k) at (\k+4.5,1.5){};
	\draw[red,thick,dashed] (zz) -- (rr\k);
}

\draw[black,thick,dotted] (4.7,1.5) -- (5.3,1.5) ;
\draw[black,thick,dotted] (5.7,1.5) -- (6.3,1.5) ;

\node[above] at (rr0) {$r_1$};
\node[above] at (rr1) {$r_i$};
\node[above] at (rr2) {$r_k$};
\node[left] at (zz) {$z$};
\node[below] at (xx) {$x$};

\draw[black,very thick] (7.3,0) -- (7.7,0) ;
\draw[black,very thick] (7.5,-0.2) -- (7.5,0.2) ;

%---- G2------

\filldraw[fill=blue!20]
  (9.5,-1.5) .. controls (8.8,-0.5) .. (9.5,0.5) .. controls (10.2,-0.5) .. (9.5,-1.5)
  -- cycle ;
\node[] at (9.5,-0.5) {$G_2$};

\node[vertex]  (yyy) at (9.5,0.5){};
\node[vertexRoot] (zzz) at (9.5,-1.5){};

\foreach \k in {0,1,2}{
	\node[vertexRoot] (rrr\k) at (\k+8.5,1.5){};
	\draw[red,thick,dashed] (yyy) -- (rrr\k);
}
\node[vertexRoot] (rrr3) at (11,1.5){};
\draw[red,thick,dashed] (yyy) -- (rrr3);

\draw[black,thick,dotted] (8.7,1.5) -- (9.3,1.5) ;
\draw[black,thick,dotted] (9.7,1.5) -- (10.3,1.5) ;
%\draw[black,thick,dotted] (10.7,1.5) -- (11.3,1.5) ;

\node[above] at (rrr0) {$r_1$};
\node[above] at (rrr1) {$r_i$};
\node[above] at (rrr2) {$r_k$};
\node[above] at (rrr3) {$x$};
\node[left] at (yyy) {$y$};
\node[below] at (zzz) {$z$};

\end{tikzpicture}
\caption{Decomposition of an extended graph resulting from a series composition.
\label{fig:series-xy}}
\end{figure}

\begin{claim} \label{cl:series-roots-xy}
$\ecost((\G^{+F\cup\{xy\}},\R),\sigma)= \ecost((\G^{+F\cup\{xy\}},\R),\sigma^*)$. 
\end{claim}

\begin{proofclaim}
The existence of the fictive edge $xy$ does not change the arguments used in the proof of Claim~\ref{cl:series-roots2}. Let us consider $v_1\in V_1$ and $v_2\in V_2\setminus\{z\}$. First as $\sigma^*(x)<\sigma^*(z)<\sigma^*(v_2)$ and $z$ separates vertices of $V_2$ from vertices of $V_1,$  we obtain that  $S^{(e)}_{\sigma^*}(v_1)\cap(V_2\setminus\{z\})=\emptyset$ and $S^{(e)}_{\sigma^*}(v_2)\cap(V_1\setminus\{x,z\})=\emptyset$. Moreover $\sigma^*[V_1\cup\R]=\sigma[V_1\cup\R]$ implies $S^{(e)}_{\sigma}(v_1)=S^{(e)}_{\sigma^*}(v_1)$ and $\sigma^*[V_2\cup\R]=\sigma[V_2\cup\R]$ implies $S^{(e)}_{\sigma}(v_2)=S^{(e)}_{\sigma^*}(v_2),$  proving the claim. In other words, $\sigma$ is a connected layout of $\G^{+F\cup\{xy\}}$ of minimum cost. 
\end{proofclaim}

%\medskip
Let us now conclude the proof of the lemma. 
Claim~\ref{cl:series-roots-xy}, the fact that $\sigma\in\mathcal{L}^c(\G^{+F\cup\{xy\}},\R)$ and the optimality of $\sigma^*$ imply that $\sigma$ is a connected layout of $\G^{+F}$ of minimum cost.
From Claim~\ref{cl:series-roots1}, we have that $\sigma_1=\sigma[V_1\cup \R]\in\mathcal{L}^c(\G[V_1\cup\R]^{+F'\cup\{xy\}},\R)$. 
{Observe that in the extended graph $\G[V_1\cup\R]^{+F'},$  for $i\in[k],$  the fictive edge $zr_i\in F'$ (for $i\in[k]$) simulates every simple extended $(z,r_i)$-path in $\G^{+F\cup\{xy\}}$. Similarly the fictive edge $zx$ aims at representing extended $(z,x)$-paths avoiding $V_1\setminus\{x\}$ in $\G^{+F\cup\{xy\}}$.} It follows that for every vertex $v_1\in V_1\cup\R,$  $S^{(e)}_{\sigma}(v_1)=S^{(e)}_{\sigma_1}(v_1)$. 
Likewise, from Claim~\ref{cl:series-roots1}, we know that $\sigma_2=\sigma[V_2\cup\R']\in\mathcal{L}^c(\G[V_2\cup\R']^{+F\cup\{zx\}},\R')$. 
 As noticed before, $z$ separates the vertices of $V_2$ from vertices of $V_1,$  and thereby for every $v_2\in V_2,$  $V_1\setminus\{z\}\cap S^{(e)}_{\sigma}(v_2)=\emptyset$. Observe that despite the fact that $x\notin V_2,$  it is preserved as a (pendant) root in $\R'$. It follows for every $v_2\in V_2\setminus\{z\},$  we have $S^{(e)}_{\sigma}(v_2)=S^{(e)}_{\sigma_2}(v_2),$  completing the proof.
\end{proof}

%---------------------------------------------------------------------------------------------------------------
\subsection{The dynamic programming algorithm}

The following upper bound on connected treewidth enables us to optimize the size of the DP tables.

\begin{theorem}[\cite{FN06a,FN06b,FN08}] \label{th:upper-bound}
Every graph $G$ on $n$ vertices satisfies $\ctw(G)\leq \tw(G)\cdot(\log n{+}1)$.
\end{theorem}

As series-parallel graphs have treewidth at most two, Theorem~\ref{th:upper-bound} implies that the connected treewidth of a series-parallel graph on $n$ vertices is at most $c_{\rm sp}=\lceil 2(\log n+1)\rceil$. This bound allows us to optimize the size of the table in our dynamic programming algorithm. Moreover this bound is tight~\cite{FN06b,FN08}, even on series-parallel graphs as certified by a construction of contraction obstructions for connected treewidth at most $k,$ for every $k\geq 2$~\cite{APT19}.
Let us first focus on the case of biconnected series-parallel graph.

\begin{proposition} \label{prop:SP}
Let $G$ be a biconnected series-parallel graph on $n$ vertices. Then computing $\ctw(G)$ can be done in $O(n^2\cdot\log n)$-time.
\end{proposition}
\begin{proof}
Let $\G=(G,(x,y)),$  with $G=(V,E),$  be a biconnected $2$-terminal graph such that $xy\in E$. By Theorem~\ref{th:SP-2terminal}, we have $\G=\G_1\oplus \G_2$ where $\G_1=(G_1,(x,y))$ with $G_1=(\{x,y\},\{xy\})$ and $\G_2=(G_2,(x,y))$ with $G_2=(V,E\setminus\{xy\})$.
By Theorem~\ref{th:SP-recognition}, in linear time, we can compute $T(\G),$  the \SPtree{} of $\G$. Recall that the root of $T(\G)$ corresponds to the parallel composition $\G_1\oplus \G_2$.
We let $\G_t=(G_t,(x_t,y_t))$ with $G_t=(V_t,E_t)$ denote the subgraph represented by node $t$ of $T(\G)$. We let  
%$\tilde{G}_t$
\xchange{$\tilde{G}_{t,k}$} be the graph $G_t$ augmented with $k$ isolated vertices $r_1,\dots, r_k$ where $k\leq c_{\rm sp}$ and denote $\tilde{\G}_{t,k}=(\tilde{G}_{t,k},(x_t,y_t))$. 
%The notation is extended to $2$-terminal graphs as well. 
In order to apply the rules described in Lemmas~\ref{lem:parallel-no-fictive} --~\ref{lem_series_fictive2}, the table $\DP_{t}[\,\cdot\,]$ stored at every node $t$ contains the following values:
%\begin{itemize}
%\item $\DP_{t}[0]=\ecvs(\G_t}^{+\emptyset},\{x_t,y_t\})$;
%\item for $k\in[c_{\rm sp}],$  $\DP_{t}[k,x_t]=\ecvs(\tilde{\G}_t[V_t\cup\R]^{+F},\R)$ where $\R=\{x_t,r_1,\dots, r_k\}$ and $F=\{y_tr_i\mid i\in[k]\}$;
%\item for $k\in[c_{\rm sp}],$  $\DP_{t}[k,y_t]=\ecvs(\tilde{\G}_t[V_t\cup\R]^{+F},\R)$ where $\R=\{y_t,r_1,\dots, r_k\}$ and $F=\{x_tr_i\mid i\in[k]\}$;
%\item for $k\in[c_{\rm sp}-1],$  $\DP_{t}[k,x_t,x_ty_t]=\ecvs(\tilde{\G}_t[V_t\cup\R]^{+F},\R)$ where $\R=\{x_t,r_1,\dots, r_k\}\}$ and $F=\{x_ty_t\}\cup\{y_tr_i\mid i\in[k]\}$;
%\item for $k\in[c_{\rm sp}-1],$  $\DP_{t}[k,y_t,x_ty_t]=\ecvs(\tilde{\G}_t[V_t\cup\R]^{+F},\R)$ where $\R=\{y_t,r_1,\dots, r_k\}\}$ and $F=\{x_ty_t\}\cup\{x_tr_i\mid i\in[k]\}$;
%\end{itemize}
\begin{itemize}
\item $\DP_{t}[0]=\ecvs(\G_t^{+\emptyset},\{x_t,y_t\})$;
\item for $k\in[c_{\rm sp}],$  $\DP_{t}[k,x_t]=\ecvs(\xchange{\tilde{\G}_{t,k}}[V_t\cup\R]^{+F},\R)$ where $\R=\{x_t,r_1,\dots, r_k\}$ and $F=\{y_tr_i\mid i\in[k]\}$;
\item for $k\in[c_{\rm sp}],$  $\DP_{t}[k,y_t]=\ecvs(\xchange{\tilde{\G}_{t,k}}[V_t\cup\R]^{+F},\R)$ where $\R=\{y_t,r_1,\dots, r_k\}$ and $F=\{x_tr_i\mid i\in[k]\}$;
\item for $k\in[c_{\rm sp}-1],$  $\DP_{t}[k,x_t,x_ty_t]=\ecvs(\xchange{\tilde{\G}_{t,k}}[V_t\cup\R]^{+F},\R)$ where $\R=\{x_t,r_1,\dots, r_k\}\}$ and $F=\{x_ty_t\}\cup\{y_tr_i\mid i\in[k]\}$;
\item for $k\in[c_{\rm sp}-1],$  $\DP_{t}[k,y_t,x_ty_t]=\ecvs(\xchange{\tilde{\G}_{t,k}}[V_t\cup\R]^{+F},\R)$ where $\R=\{y_t,r_1,\dots, r_k\}\}$ and $F=\{x_ty_t\}\cup\{x_tr_i\mid i\in[k]\}$;
\end{itemize}

The bounds on the integer $k$, determining the number of entries in the table $\DP_t[\,\cdot\,]$ of a node $t$ is, are delimited by the upper-bound $c_{\rm sp}$, as asserted by Theorem~\ref{th:upper-bound}.
The initialization of the table for leaf nodes (see below) guarantees  that this bound is respected.

We observe that for every node $t,$  every entry of $\DP_t[\,\cdot\,]$ corresponds to an extended rooted two-terminal graph $(\mathbf{H}^{+F},\R)$ such that: $\R$ contains at least two vertices; at least one vertex of $\R$ is a terminal vertex; and every root vertex that is not a terminal vertex is an isolated vertex. These properties implies that every connected component of $\mathbf{H}^{+F}$ contains a root vertex, and thereby it guarantees the existence of a connected layout of $(\mathbf{H}^{+F},\R)$.

Suppose that $t$ represents a parallel composition $\G_t=\G'_1\oplus \G'_2$. The children $t_1$ and $t_2$ of $t$ respectively represent the $2$-terminal graphs $\G'_1=(G'_1,(x_t,y_t))$ and $\G'_2=(G'_2,(x_t,y_t))$. Then $\DP_t[\,\cdot\,]$ is computed as follows:

\begin{itemize}
\item By Lemma~\ref{lem:parallel-no-fictive}, $\DP_{t}[0]=\max \left\{\DP_{t_1}[0],\DP_{t_2}[0]\right\}$.
\item By Lemma~\ref{lem_parallel_fictive}, we have for $k\in[c_{\rm sp}-1]$:
$$\DP_{t}[k,x_t]=\min \left\{
\begin{array}{c} 
\max \left\{\DP_{t_1}[k,x_t,x_ty_t], \DP_{t_2}[0]\right\}\\ 
\\
\max \left\{\DP_{t_2}[k,x_t,x_ty_t], \DP_{t_1}[0] \right\} \\
\end{array}\right\} \mbox{ and}
$$

$$\DP_{t}[k,y_t]=\min \left\{
\begin{array}{c} 
\max \left\{\DP_{t_1}[k,y_t,x_ty_t], \DP_{t_2}[0]\right\}\\ 
\\
\max \left\{\DP_{t_2}[k,y_t,x_ty_t], \DP_{t_1}[0] \right\} \\
\end{array}\right\}.
$$

\item By Lemma~\ref{lem_parallel_fictivebis}, we have for $k\in[c_{\rm sp}-1]$: $\DP_{t}[k,x_t,x_ty_t]=\DP_{t}[k,x_t]$ and $\DP_{t}[k,y_t,x_ty_t]=\DP_{t}[k,y_t]$.
\end{itemize}

\medskip
Suppose that $t$ represents a series composition $\G_t=\G'_1\otimes \G'_2$. The children $t_1$ and  $t_2$ of $t$ respectively represent the $2$-terminal graphs $\G'_1=(G'_1,(x_t,z))$ and $\G'_2=(G'_2,(z,y_t))$. Then $\DP_t[\,\cdot\,]$ is computed as follows:

\begin{itemize}
\item By Lemma~\ref{lem_series_no_fictive}, we have for $k\in[c_{\rm sp}-1]$:
$$\DP_{t}[0]=
\min \left\{
\begin{array}{c} 
\max \left\{\DP_{t_1}[x_t],\DP_{t_2}[0]\right\}\\
\\
\max \left\{\DP_{t_2}[y_t],\DP_{t_1}[0]\right\}\\
\end{array}
\right\}.
$$

\item By Lemma~\ref{lem_series_fictive}, we have for $k\in[c_{\rm sp}-1]$: 
$$\DP_{t}[k,x_t]=\max \left\{\DP_{t_1}[k,x_t], \DP_{t_2}[k,z]\right\} \mbox{ and}$$
$$\DP_{t}[k,y_t]=\max \left\{\DP_{t_1}[k,y_t], \DP_{t_2}[k,z]\right\}.$$

\item By Lemma~\ref{lem_series_fictive2}, we have for $k\in[c_{\rm sp}-2]$:

$$\DP_{t}[k,x_t,x_ty_t]=\max \left\{\DP_{t_1}[k,x_t,x_tz], \DP_{t_2}[k+1,z]\right\} \mbox{ and}$$
$$\DP_{t}[k,y_t,x_ty_t]=\max \left\{\DP_{t_1}[k,y_t,y_tz], \DP_{t_2}[k+1,z]\right\}.$$

\end{itemize}

For every non-leaf node $t$ of $T(\G),$  the entries of $\DP_t[\,\cdot\,]$ are initialized to some dummy value $\bot$. Every leaf node $t$ of $T(\G)$ represents the single edge graph, that is $G_t=(V_t,E_t)$ with $V=\{x_t,y_t\}$ and $E=\{x_ty_t\}$. 
%We let denote 
%$\tilde{G_t}$ 
%$\xchange{\tilde{G}_{t,k}}$
%the graph $G_t$ augmented with $k$ isolated vertices $r_1,\dots, r_k$. 
Then we can initialize the values associated to a leaf node $t$ as follows:
\begin{itemize}
\item $\DP_{t}[0]=\ecvs(\G_t^{+\emptyset},\{x_t,y_t\})=1$
\item for $k\in[c_{\rm sp}-1],$  $\DP_{t}[k,x_t]=\ecvs(\xchange{\tilde{\G}^{+F}_{t,k}},\{x_t,r_1,\dots, r_k\})={k+1}$ where $F=\{y_tr_i\mid i\in[k]\}$.
\item for $k\in[c_{\rm sp}-1],$  $\DP_{t}[k,y_t]=\ecvs(\xchange{\tilde{\G}^{+F}_{t,k}},\{y_t,r_1,\dots, r_k\})={k+1}$ where $F=\{x_tr_i\mid i\in[k]\}$.
\item for $k\in[c_{\rm sp}-2],$  $\DP_{t}[k,x_t,x_ty_t]=\ecvs(\xchange{\tilde{\G}^{+F}_{t,k}},\{x_t,r_1,\dots, r_k\})={k+1}$ where $F=\{y_tr_i\mid i\in[k]\}\cup\{x_ty_t\}$.
\item for $k\in[c_{\rm sp}-2],$  $\DP_{t}[k,y_t,x_ty_t]=\ecvs(\xchange{\tilde{\G}^{+F}_{t,k}},\{y_t,r_1,\dots, r_k\})={k+1}$ where $F=\{x_tr_i\mid i\in[k]\}\cup\{x_ty_t\}$.
\end{itemize}

As the \SPtree{} $T(\G)$ contains $O(n)$ nodes, filling the table $\DP_t[\,\cdot\,]$,  for every node $t,$  is achieved in $O(n\cdot \log n)$-time. Theorem~\ref{th:equiv} states that $\ctw(G)=\cvs(G)=\min\{\ecvs(\G^{+\emptyset},\{x,y\})\mid xy\in E\}$. This implies that the whole algorithm runs in $O(n^2\cdot \log n)$-time.
\end{proof}

%---------------------------------------------------------------------------------------------------------------
\subsection{Generalization to graph of treewidth at most two}

Recall that a graph $G$ has treewidth at most two if and only if every biconnected component of $G$ induces a series-parallel graph. So we need a lemma to deal with cut vertices.

\begin{lemma} \label{lem:cutvertex}
Let $G=(V,E)$ be a graph containing a cut vertex $x$ and let $G_1=[C_1\cup\{x\}],$  \dots, $G_k=G[C_k\cup\{x\}]$ be the induced subgraphs where $C_1, \dots, C_k$ denote the connected components of $G-x$. Then 
\[
\cvs(G)
=\min_{i\in[k]} \Bigl\{ \max\left\{\cvs(G_i),\max\left\{\cvs(G_j,\{x\})\mid j\in [k], j\neq i\right\}\Bigr\} 
\right\} \mbox{ (see Figure~\ref{fig:cut-vertex})}.
\]
\end{lemma}

\begin{figure}[h]
\centering
\begin{tikzpicture}[scale=1.0]
   \tikzstyle{vertex}=[fill,circle,minimum size=0.2cm,inner sep=0pt]
   \tikzstyle{vertexRoot}=[fill,red,rectangle,minimum size=0.15cm,inner sep=0pt]

%---- G------

\begin{scope}
\coordinate (A) at (0,0);
\coordinate (B) at (2,0);
\coordinate (S) at (1,0.8);
\coordinate (T) at (2,1);
\coordinate (SS) at (1,-0.8);
\coordinate (TT) at (2,-1);

\filldraw[fill=red!20] 
(A) .. controls (S) and (T) .. (B)  .. controls (TT) and (SS) .. (A)
-- cycle;

\node[] at (1,0) {$G_3$};
\end{scope}

\begin{scope}[rotate=120]
\coordinate (A) at (0,0);
\coordinate (B) at (2,0);
\coordinate (S) at (1,0.8);
\coordinate (T) at (2,1);
\coordinate (SS) at (1,-0.8);
\coordinate (TT) at (2,-1);

\filldraw[fill=blue!20] (A) .. controls (S) and (T) .. (B)  .. controls (TT) and (SS) .. (A)
-- cycle;

\node[] at (1,0) {$G_2$};
\end{scope}

\begin{scope}[rotate=240]
\coordinate (A) at (0,0);
\coordinate (B) at (2,0);
\coordinate (S) at (1,0.8);
\coordinate (T) at (2,1);
\coordinate (SS) at (1,-0.8);
\coordinate (TT) at (2,-1);

\filldraw[fill=black!20] (A) .. controls (S) and (T) .. (B)  .. controls (TT) and (SS) .. (A)
-- cycle;

\node[] at (1,0) {$G_1$};
\end{scope}

\node[vertex]  (x) at (0,0){};
\node[left] at (-0.2,0) {$x$};

\draw[black,very thick,->] (2.5,0) -- (3,0) ;

%---- G1------
%shift={(5,-1)},rotate=-45

\begin{scope}[shift={(5,0.85)},rotate=240]
\coordinate (A) at (0,0);
\coordinate (B) at (2,0);
\coordinate (S) at (1,0.8);
\coordinate (T) at (2,1);
\coordinate (SS) at (1,-0.8);
\coordinate (TT) at (2,-1);

\filldraw[fill=black!20] (A) .. controls (S) and (T) .. (B)  .. controls (TT) and (SS) .. (A)
-- cycle;

\node[] at (1,0) {$G_1$};

\node[vertex]  (x) at (0,0){};
\node[right] at (x) {$x$};
\node[vertexRoot]  (r) at (1.5,0.4){};

\end{scope}

\draw[black,very thick] (5.6,0) -- (6,0) ;
\draw[black,very thick] (5.8,0.2) -- (5.8,-0.2) ;

%---- G2------

\begin{scope}[shift={(8,-0.85)},rotate=120]
\coordinate (A) at (0,0);
\coordinate (B) at (2,0);
\coordinate (S) at (1,0.8);
\coordinate (T) at (2,1);
\coordinate (SS) at (1,-0.8);
\coordinate (TT) at (2,-1);

\filldraw[fill=blue!20] (A) .. controls (S) and (T) .. (B)  .. controls (TT) and (SS) .. (A)
-- cycle;

\node[] at (1,0) {$G_2$};

\node[vertexRoot]  (x) at (0,0){};
\node[right] at (x) {$x$};

\end{scope}

\draw[black,very thick] (8.6,0) -- (9,0) ;
\draw[black,very thick] (8.8,0.2) -- (8.8,-0.2) ;

%---- G3------

\begin{scope}[shift={(9.7,0)}]
\coordinate (A) at (0,0);
\coordinate (B) at (2,0);
\coordinate (S) at (1,0.8);
\coordinate (T) at (2,1);
\coordinate (SS) at (1,-0.8);
\coordinate (TT) at (2,-1);

\filldraw[fill=red!20] 
(A) .. controls (S) and (T) .. (B)  .. controls (TT) and (SS) .. (A)
-- cycle;

\node[] at (1,0) {$G_3$};

\node[vertexRoot]  (x) at (0,0){};
\node[left] at (x) {$x$};
\end{scope}

\end{tikzpicture}
\caption{Decomposition of a graph with a cut vertex. If an optimal connected layout $\sigma^*$ starts at an arbitrary vertex of $G_1=G[C_1\cup\{x\}]$, then $\sigma^*[C_2]$ and $\sigma^*[C_3]$ start at $x$, which becomes a root of $G_2=G[C_2\cup\{x\}]$ and of $G_3=G[C_3\cup\{x\}]$.
\label{fig:cut-vertex}}
\end{figure}
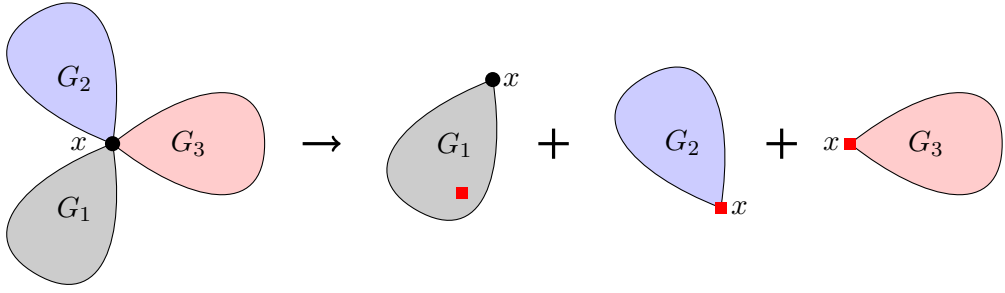
\begin{proof}
Let us consider $\sigma\in\mathcal{L}^{(c)}(G)$ and suppose that the first vertex of $\sigma$ belongs to $C_1$. Then observe that $\sigma$ can be rearranged into $\tau=\sigma[C_1]\odot \sigma[C_2\setminus\{x\}]\dots, \odot\sigma[C_k\setminus\{x\}]$ and that since $x$ is a cut vertex then $\tau\in\mathcal{L}^{(c)}(G)$ as well.
\xchange{The statement follows from the observation that $\cost(G,\tau)=\max\left\{\cost(G_i,\sigma[C_i])\mid i\in [k]\right\}$.}
\end{proof}

\begin{theorem}
Computing the connected treewidth of a graph of treewidth at most $2$ requires $O({n^3\cdot \log n})$-time.
\end{theorem}
\begin{proof}
Let $G$ be a graph of treewidth at most $2$. The algorithm first computes the biconnected tree decomposition of $G$. This can be done in linear time.
Following Lemma~\ref{lem:cutvertex}, we guess a biconnected component $C_1$ in which the connected layout will start. This generates for every biconnected component $C_k$ distinct from $C_1$ a root vertex $r_k$. Then using Proposition~\ref{prop:SP}, in $O(n^2\cdot\log n)$ we can compute $\cvs(G[C_1])$ and $\cvs(G[C_k],\{r_k\})$ for each $k\neq 1$. This leads to an $O(n^3\cdot \log n)$-time algorithm.
\end{proof}

%---------------------------------------------------------------------------------------------------------------
%---------------------------------------------------------------------------------------------------------------
\section{Discussion and open problems}
\label{klopath}

We obtained a polynomial time algorithm to compute the connected treewidth for the class of graphs of treewidth at most two. This result naturally leads to the problem of determining the  algorithmic complexity of computing the connected treewidth for the class of  bounded treewidth graphs. To discuss this, we present the problem as a decision  problem:\medskip

\noindent {\sc Connected Treewidth}\\
\noindent {\sl Input:} A graph $G$ and an integer $k$.\\
\noindent {\sl Question:} $\ctw(G)\leq k$?
\medskip

Our result implies  that {\sc Connected Treewidth} can be solved in $O(n^3\cdot k)$-time for graphs of treewidth at most $2$. 
Let us discuss the following three conjectures.

\begin{conjecture}
\label{conj1}
{\sc Connected Treewidth} can be solved by an $O(n^{f(\tw(G))})$-time algorithm.
\end{conjecture}

\begin{conjecture}
\label{conj2}
{\sc Connected Treewidth} can be solved by an   $O(f(k,\tw(G))\cdot n)$-time algorithm.
\end{conjecture}

\begin{conjecture}
\label{conj3}
{\sc Connected Treewidth} can be solved by an   $O(n^{f(k)})$-time algorithm.
\end{conjecture}

Our result can be seen as a special case of Conjecture \ref{conj1} (when $\tw(G)\leq 2$). A general resolution of Conjecture \ref{conj1} would require
a vast extension of 
our dynamic programming approach.
In our approach (for treewidth at most two) we essentially solve a slightly modified problem where the input is a pair $(G,e),$ 
 where $e\in E(G),$  and we return the minimum cost of a layout 
 that {\sl begins} with the endpoints of $e$. Then we reduce   the computation of connected treewidth to this problem by paying an  overhead of $O(n^2)$. An interesting question is whether and how a similar approach might work 
 for the general case. 
 Of course, one may try to reduce the  $O(n^3\cdot \log n)$-time complexity 
 of our algorithm by avoiding such reductions and directly build a dynamic programming scheme for {\sc Connected Treewidth} on graphs of treewidth $\leq k$. We believe that this is possible and can  reduce the time complexity to $O(n^c)$ for some $1<c\leq 3$.

For  Conjecture \ref{conj2} one may attempt to use tools related to Courcelle's theorem. This would require to express the question $\ctw(G)\leq k$ in using a formula $\phi_{k}$ in Monadic Second Order Logic (MSOL) which is far from being obvious. A possible direction would be to consider the contraction-obstruction set ${\cal Z}_{k}$ of the class ${\cal G}_{k}=\{G\mid \ctw(G)\leq k\},$  i.e., the contraction minimal graphs not in ${\cal G}_{k}$. Indeed contraction testing is MSOL expressible. 
However, it turns out  that, unlike the case for treewidth and pathwidth with respect to minors, ${\cal Z}_{k}$ is infinite for every $k\geq 2,$  as observed in~\cite{APT19}. A possible 
way to overcome this obstacle is to consider some other partial ordering relation, alternative to contractions, that maintains closeness, MSOL expressibility
and gives rise to bounded size obstructions. Such a step can be done using the results of~\cite{APT19} for graphs of connected treewidth at most two. However, it is not clear whether this can be extended for bigger values of connected treewidth.
Let us mention that recently, Kant\'e et al.~\cite{Kante2020linear} obtained a $O(f(k)\cdot n)$-time algorithm to compute the connected pathwidth of a graph (that is the analogue of Conjecture~\ref{conj2} for pathwidth). {Finally, it is also natural to ask if one can compute in {\sf FPT}-time the connected treewidth of a graph when parameterized by the pathwidth of the input graph.}

A proof of Conjecture \ref{conj3} would follow if we devise an algorithm to check whether {\sf ctw}$(G)\leq k$  in   $O(n^{f(k,\tw(G))})$  time. This follows 
directly from the fact that {\sf yes}-instances of {\sc Connected Treewidth} have always treewidth at most $k$. Such a result would be analogous to 
the one of~\cite{DOR19} for connected pathwidth and is perhaps the first  (and easier) to be attacked among the three above conjectures.

%%---------------------------------------------------------------------------------------------------------------
%%---------------------------------------------------------------------------------------------------------------
%\bibliographystyle{abbrv}
%\bibliography{SP-connected-tw.bib}

\end{document}